\documentclass[a4paper,12pt]{article}
\usepackage{graphicx}
\usepackage{multirow}
\usepackage{bbm}
\usepackage[labelfont=bf,font=normalsize]{caption}
\usepackage{float}
\usepackage{footnote}
\usepackage{amssymb}
\usepackage{pifont}
\usepackage{ifpdf}
\ifpdf
\else
\usepackage{pstcol,pst-fill,pstricks}
\fi
\usepackage{natbib}
\usepackage{mathpazo}
\usepackage{import}
\usepackage{dsfont}
\usepackage{enumerate}
\usepackage{lscape}
\usepackage{epsfig}
\usepackage[latin1]{inputenc}
\usepackage[OT1]{fontenc}
\usepackage[T1]{fontenc}

\usepackage[margin=1in]{geometry}
\usepackage[doublespacing]{setspace}

\usepackage{color}
\usepackage{mdframed}
\usepackage{tikz}
\usepackage{tkz-graph}
\usepackage[bottom,flushmargin]{footmisc}
\usepackage{subcaption}  
\usepackage{times}
\usepackage{fancyhdr,graphicx,amsmath,amssymb}
\usepackage[ruled,vlined]{algorithm2e}
\usepackage{mathtools}
\usepackage{amsmath}
\usepackage{amsthm}
\usepackage{thmtools}
\usepackage{calrsfs}
\usepackage{booktabs}

\makeatletter
\newcommand*\bigcdot{\mathpalette\bigcdot@{.5}}
\newcommand*\bigcdot@[2]{\mathbin{\vcenter{\hbox{\scalebox{#2}{$\m@th#1\bullet$}}}}}
\makeatother
\usepackage{authblk}
\usepackage{setspace}
 \newcommand{\be}{\begin{equation}}
\newcommand{\ee}{\end{equation}}

\newcommand{\abar}{\underline{a}}

\DeclareMathAlphabet{\pazocal}{OMS}{zplm}{m}{n}
\allowdisplaybreaks

\usepackage[hyperfootnotes=false]{hyperref}
 \hypersetup{colorlinks=true,
 linkcolor=blue,
 citecolor=blue,
 filecolor=black,
 urlcolor=black,
 pdfstartview={Fit},
pdfpagemode=UseNone
 }
\usepackage{cleveref}

\declaretheoremstyle[
  headfont=\normalfont\scshape,
  numbered=unless unique,
  bodyfont=\normalfont,
  spaceabove=1em,
  spacebelow=1em,
]{exmpstyle}

\newcommand{\p}{\partial}

\newcommand{\defeq}{\vcentcolon=}

\def\ee{\mathsf{e}}

\usepackage[mathlines]{lineno}

\newtheorem{Ass}{Assumption}
\newtheorem{prop}{Proposition}

\theoremstyle{definition}

\newtheorem{theorem}{Theorem}
\newtheorem{corollary}{Corollary}

\title{The Response of Consumption to Interest Rates with Borrowing Constraints: An Analytical Approach\thanks{I thank Benjamin Moll, Bar Light as well as participants at the NUS brownbag seminar for their interesting comments.}}
\author[1]{Jordan Roulleau-Pasdeloup}
\affil[1]{Banque de France}
\date{\today}

\begin{document}
\begin{titlepage}
\maketitle
\thispagestyle{empty}

\abstract{I derive an explicit mapping from initial assets, income and the real interest rate to consumption for an income fluctuation problem with a borrowing constraint and CARA utility. I show that there exists a threshold of initial wealth over which the partial equilibrium consumption response to a permanent increase in the real interest rate is positive, consistent with recent empirical evidence. I further show that precautionary savings reinforce the possibility of crowding in and that the possibility of a positive response extends to a more standard CRRA utility whenever the elasticity of intertemporal substitution is strictly less than 1.}
\vspace{.5cm}
\end{titlepage}


\section{Introduction}

What are the drivers of consumption? This is a question that is almost as old as formal macroeconomics itself. Indeed, the dynamics of consumption feature front and center in the early contributions of \cite{Ramsey1928mathematical}, \cite{Fisher1930theory} and \cite{Keynes1936general}. The approaches differ in the sense that both \cite{Ramsey1928mathematical} and \cite{Fisher1930theory} derive consumption decisions from optimization while \cite{Keynes1936general} postulates a 'reasonable' consumption function. Another difference in approach lies in the fact that, while the first two placed a large emphasis on assets and interest rates, \cite{Keynes1936general} argued that the main determinant of current consumption was current disposable income. The recent and growing literature on heterogeneous agents New Keynesian (HANK) models \textemdash a terminology introduced in \cite{Kaplan2018monetary}\textemdash unifies these approaches and finds that disposable income plays a key role in optimal consumption decisions alongside interest rates and assets all while taking borrowing constraints seriously following \cite{Schechtman1976income}. Despite decades of work on the topic however, we don't have an explicit formula for how initial assets, income and interest rates affect consumption decisions in the presence of borrowing constraints.

In this context, the main contribution of this paper is to solve explicitly for the optimal consumption decision in closed form for an infinitely-lived household facing both borrowing constraints as well as stochastic income. The optimal consumption decision obtained in this framework is expressed explicitly as a function of initial assets, income and interest rates. The contribution borrows from and contributes to two growing literatures: a macro literature that focuses on the impact of interest rates on consumption decisions and a micro literature that focuses on initial assets and income as the primary drivers of said decisions. I use this framework to study explicitly how changes in real interest rates interact with initial asset positions to shape consumption decisions. This approach is thus largely complementary to the numerical solutions that allow us to quantify the effects in that it informs us about the role of different structural parameters. Before delving into the specifics of the approach that permit a closed form characterization, I briefly describe these two literatures. 

From a macroeconomics perspective, recent contributions such as \cite{Kaplan2018monetary} as well as \cite{Bilbiie2025hanksson} have emphasized the role of monetary policy through interest rates as well fiscal policy through taxes and redistribution in shaping consumption.\footnote{See also \cite{Debortoli2024idiosyncratic}.} At the same time, the predictions of these theories can now be explored in great detail thanks to the recent access to high quality
micro-data. Two recent contributions include \cite{Holm2021transmission} and \cite{Bilbiie2025hanksson}, who both use high-quality administrative data from Norway. Another evidence for the enduring legacy of  \cite{Keynes1936general}'s approach is that his concept of a \textit{consumption function} is a technical apparatus that is often used to frame discussions in the recent literature \textemdash see \cite{Bilbiie2020new}. The consumption functions that can be found in the HANK literature are virtually all computed numerically, see for recent example \cite{Kase2025estimating} who use a Deep Learning approach. That approach highlights that both liquid and illiquid assets are important drivers of consumption decisions. The approach in \cite{Bilbiie2020new} and \cite{Bilbiie2025hanksson} is different and focuses on analytical tractability. This offers new perspectives and allows one to gain valuable intuition about the differing behaviors of hand-to-mouth households versus savers. In that approach, the consumption function is usually computed as a deviation from a steady state in which assets are in zero net supply. As a result, interest rates become the main drivers and the consumption of savers is heavily influenced by inter-temporal substitution motives. 

From a micro perspective, there is a growing literature that dates back to \cite{Schechtman1976income} and tries to characterize the optimal consumption function in the presence of borrowing constraints and stochastic income. Except for very few special cases such as the one in \cite{Hall1978stochastic} where certainty equivalence holds, it quickly became apparent (see \cite{Zeldes1989optimal}) that optimal consumption functions had to be solved numerically. Against this backdrop, the recent literature on this topic has then focused on establishing the \textit{properties} of the consumption function.\footnote{We know from \cite{Carroll1996concavity} that, for utility functions with hyperbolic absolute risk aversion (HARA), the consumption function has to be \textit{concave} in initial assets. In a recent contribution, \cite{Toda2021necessity} shows that the HARA assumption is necessary for concavity. We know from \cite{Benhabib2015wealth} that the consumption function is \textit{asymptotically linear} in initial assets. See also \cite{Gouin2023pareto}.} Some progress has been made on deriving closed-form solutions in these environments in \cite{Park2006analytical}, \cite{Holm2018consumption}, \cite{Achdou2022income}, and more recently \cite{Roulleau2026explicit} in continuous-time settings. While some contributions in this literature such as \cite{Ma2020income}, \cite{Commault2025Permanent_income} and \cite{Roulleau2026explicit} offer some results on income, by far most of the focus has been on initial assets as a driver of consumption.\footnote{One notable exception is \cite{Lehrer2018effect}, who study the effect of interest rates in a typical discrete-time income fluctuation problem. They do not derive optimal consumption as an explicit function of interest rates which is the main focus of this paper.}  

As stated before, the goal of this paper is to try to bridge the gap between these two literatures and derive an optimal consumption function where the roles of initial assets, income and interest rates are made fully explicit. The two key ingredients that make this characterization possible are $(i)$ continuous time and $(ii)$ a constant absolute risk aversion (CARA) utility function. The results obtained in \cite{Helpman1981optimal} imply that consumption functions in discrete time income fluctuation problems are piece-wise linear and can only be expressed implicitly. Loosely speaking, the assumption of continuous time smooths out a piece-wise linear function into one that can be expressed analytically. The CARA assumption has a storied history in macroeconomics and finance dating back to \cite{Merton1971optimum}, \cite{Kimball1989precautionary} as well as \cite{Caballero1990consumption}. In these frameworks, it considerably simplifies the algebra given that the consumption Euler equation implies that the derivative (instead of the growth rate with a CRRA utility function) of consumption is constant across time. A complementary simplification arises from the absence of wealth effects under CARA which makes optimal consumption linear in assets and thus makes the model amenable to aggregation in a context with heterogeneous agents.\footnote{See \cite{Wang2003caballero}, \cite{Wang2007equilibrium}, \cite{Toda2017huggett}, \cite{Acharya2020understanding} and \cite{Acharya2023optimal} for examples leveraging this property.}

In contrast with these contributions, I show that, in the presence of a borrowing constraint, consumption will \textit{not be a linear function} of initial assets but a \textit{concave} one instead. In that sense, the consumption function inherits a property that has by now become standard. This means that in the current framework the marginal propensity to consume (henceforth MPC) out of initial assets will be a non-trivial function of initial assets and not a counterfactually constant one as in the existing literature. In addition, it is well-known that consumption can turn negative with a CARA utility function because the (lower) \cite{Inada1963two} condition fails \textemdash the marginal utility of consumption at 0 is finite, see \cite{Blanchard1988consumption}. This however doesn't happen in the framework developed in this paper, once again because of the presence of an explicit borrowing constraint. In that sense, the current framework leverages the simplicity brought about by the CARA utility function without having any of its well-known shortcomings. 

As a useful baseline, I start with a framework in which there is no borrowing constraint. I derive the optimal consumption decision as an explicit function of assets, income and the real interest rate. I show that the MPC out of an increase in the real interest rate is shaped by inter-temporal substitution and income effects. For a low level of initial assets, the inter-temporal substitution effect dominates and consumption decreases after an increase in real interest rates. This is the typical behavior of a saver in a two agents new Keynesian (TANK) model \`a la \cite{Bilbiie2008limited}. I show that there is a threshold level of initial assets after which the income effect dominates and consumption \textit{increases} after an increase in real interest rates. This is very much in line with the empirical evidence reported in \cite{Holm2021transmission}. More generally, the current paper offers a formalization of the impact on consumption of the savings redistribution channel emphasized by \cite{Doepke2006inflation} in the context of inflation, by \cite{Coibion2017innocent} for monetary policy shocks and by \cite{Di2017interest} for mortgage rates: changes in real interest rates have heterogeneous effects on consumption depending on households' net asset positions.

Moving on to a setup with a borrowing constraint, I show that one can use the (first branch of the) Lambert W function to write consumption in terms of assets, income and the real interest rate explicitly. In turn, I leverage this expression to show that the MPC out of an increase in the real interest rate is larger in level compared to the unconstrained case: in the constrained case, when households cut back on consumption they do it by less and when they increase consumption they do it by more compared to the unconstrained case. With low initial assets, the risk of hitting the constraint becomes higher so that households engage less in inter-temporal substitution and cut back less on consumption. With high initial assets, households have a longer runway to deplete their assets and thus consume comparatively more out of an increase in real interest rates. These two effects together imply that the threshold level of initial assets above which the MPC out of an increase in the real interest rate is positive is \textit{lower}: consumption crowding-in becomes more likely. 

The CARA assumption allows one to get one step further and maintain an explicit solution in the presence of stochastic income. To make this possible, I develop an extension that is conceptually close in spirit to \cite{Kimball1990precautionary}. More precisely, I assume that income follows a two-state Poisson process that takes on values $y+\sigma$ and $y-\sigma$ with $\sigma\geq 0$. The key assumption is that $\sigma$ is small. Under that assumption, one can derive an optimal consumption function that features a precautionary motive. I show that with low initial assets, an increase in interest rates makes consumption decrease by more: high interest rates decrease the permanent value of human wealth conditional on being a low income type. Therefore, low initial wealth households have more incentives to save for precautionary reasons and cut back on consumption by more. At the other extreme of the initial asset distribution, an increase in interest rates increases the value of existing precautionary savings. Therefore, for a given downward income risk high initial wealth households need to save less and thus consume more. 

In this context, the use of a CARA utility function in this paper should be understood as a tractability device that isolates the interaction between borrowing constraints and interest rate effects as it allows for analytical results that are otherwise unavailable in a more standard CRRA setting. That being said, one would still like to know whether the crowding-in result is specific to a CARA utility function and whether it also survives under the assumption of a more standard CRRA utility function. This is not clear cut as \cite{Lehrer2018effect} show that when substitution effects dominate (when the Elasticity of Intertemporal Substitution (EIS) is higher than 1), consumption is crowded out after an increase in the real interest rate. I complement their analysis and show that for the arguably more relevant case where the EIS is lower than 1, there exists a threshold level of assets above which consumption is crowded in after an increase in the real interest rate. While one cannot express the consumption function in closed form in that case, one can still use the properties of the Lambert W function to study the sign of the MPC with respect to the interest rate.

\textbf{Related literature}\textemdash Given its focus on solving a consumption/savings problem with borrowing constraints and/or stochastic income, this paper is related to \cite{Schechtman1976income}, \cite{Schechtman1977some}, \cite{Helpman1981optimal}, \cite{Deaton1989saving}, \cite{Zeldes1989optimal}, \cite{Kimball1990precautionary}, \cite{Carroll1996concavity}, \cite{Wang2003caballero}, \cite{Benhabib2015wealth}, \cite{Toda2017huggett}, \cite{Lehrer2018effect}, \cite{Ma2020income}, \cite{Ma2021theory}, \cite{Toda2021necessity}, \cite{Ma2022asymptotic}, \cite{Gouin2023pareto}, \cite{Ma2025theory}, \cite{Commault2025Permanent_income}, \cite{Lee2025wealth}, \cite{Ma2026optimal} and \cite{Toda2026characterization} in discrete time as well as \cite{Park2006analytical}, \cite{Wang2007equilibrium}, \cite{Holm2018consumption}, \cite{Achdou2022income}, \cite{Fischer2026concave} and \cite{Roulleau2026explicit} in continuous time. These papers either give an \textit{implicit} consumption function, an inverse consumption function where assets are expressed as a function of consumption, or an asymptotic expression for high/low initial assets. In \cite{Roulleau2026explicit}, I derive an explicit consumption function with a CRRA utility function which is exact when $r=0$ and is only an approximation when $r\neq 0$. As a result, it is not well suited for the study of changes in interest rates on consumption.

This paper is also related to the macroeconomics literature on the effects of changes in interest rates on consumption such as \cite{Clarida1999science}, \cite{Kaplan2018monetary}, \cite{Bilbiie2020new}, \cite{Debortoli2024idiosyncratic}, \cite{Bilbiie2025hanksson} and many others that use a dynamic stochastic general equilibrium structure to answer these questions. The results reported in the current paper are also related to and complementary with the sufficient statistics approach used in \cite{Auclert2019monetary} and \cite{Farhi2022price}. While they derive formulas for the effect of changes in current/future interest rates on consumption using sufficient statistics that can be computed from the data, I derive formulas for how this effect changes with structural parameters from the model. Note that the macro literature usually focuses on persistent but transitory deviations of interest rates from their baseline. In this sense, the explicit results in the current paper can be viewed as a limit case where these deviations are highly persistent.

Finally, this paper is also closely related to the literature that seeks to understand empirically the effects of unexpected changes in interest rates from a macro perspective as in \cite{Romer2004new}, \cite{Gertler2015monetary}, \cite{Nakamura2018high}, \cite{Miranda2021transmission}, \cite{Bauer2023alternative} or from a micro perspective as in \cite{Di2017interest}, \cite{Cloyne2020monetary}, \cite{Holm2021transmission}, \cite{Eichenbaum2022state} and more recently \cite{Foulis2026interest}.

The rest of the paper is organized as follows. In Section \ref{sec:cons_func} I derive an explicit closed-form solution for a consumption/savings problem with CARA utility, borrowing constraints and constant permanent income. In Section \ref{sec:cons_rates}, I study the reaction of consumption to an increase in the real interest rate and derive an explicit threshold for that reaction to flip signs. In Section \ref{sec:extensions}, I study how that threshold changes across two extensions: one with stochastic income and one where the utility function is CRRA. I conclude in Section \ref{sec:conclusion}.

\section{An explicit consumption function and its properties}
\label{sec:cons_func}

In order to set the stage, I start with a consumption savings problem in which income is constant over time. I later extend the model to a setting with stochastic income that nests the current one as a special case.

\subsection{Preferences and technology}

Consider the following household maximization program in continuous time:
\begin{align}
\notag
\max_{\{c_t\}} &\int_0^{\infty} e^{-\rho t} u(c_t) dt\\
\label{eq:max_pgm}
\text{s.t}\quad \frac{da_t}{dt} &= r a_t + y - c_t\\
\qquad c_t &\geq 0\ ,\ a_t\geq \underline{a}
\notag
\end{align}
where $r$ is the real interest rate, $y$ is (for now, constant) permanent income and $\underline{a}$ is the borrowing limit for this household. The utility function is of the CARA form with $u(c) = -(e^{-\gamma c})/\gamma$. Throughout the paper, I maintain the following assumption:
\begin{Ass}
\label{ass:parameters}
The parameters are such that $\rho>r$, $\gamma>0$ and $\abar<0$.    
\end{Ass}

\subsection{The optimal consumption function}
\label{sec:optimal_path}

As in \cite{Holm2018consumption}, the value function for this problem solves the following Hamilton-Jacobi-Bellman equation:
\begin{align}
\rho V(a) = \max_{c\geq 0} \left\{u(c) + V'(a)\cdot(ra+y-c)\right\}   
\end{align}
subject to $a\geq \abar$. Taking the first order conditions with respect to $c,a$ and using the envelope condition one obtains after re-arranging:
\begin{align}
c(a) = ra+y + \frac{\rho-r}{\gamma c'(a)}
\label{eq:ODE}
\end{align}
which is a special case of equation (7) in \cite{Holm2018consumption} who works with a more general class of utility functions. Under this more general class, he shows that this first order non-linear ODE is solvable and provides an implicit expression for the consumption function. In contrast, the specific utility function considered here is amenable to an explicit closed form solution. This solution is described in the following theorem:
\begin{theorem}
\label{thm:cons_func}
Assume that the Assumption \ref{ass:parameters} holds and that $\abar>-y/r$. It follows that the optimal consumption function that solves equation \eqref{eq:ODE} is given by:
\begin{align}
\label{eq:cstar_W0}
c^\ast(a,y;r,\abar) &= r a+y+\frac{\rho-r}{\gamma r}\left[1+W_0\left(f(a)\right)\right]\\
f(a)& \defeq -e^{-\left[1+(a-\abar)\frac{\gamma r^2}{\rho-r}\right]}
\end{align}
where $a\in(\abar,\infty)$ and $W_0$ denotes the first branch of the Lambert W function. Given the properties of this function, it follows that:
\begin{itemize}
    \item $W_0\left(f(a)\right)\to -1$ as $a\to\abar$
    \item $W_0\left(f(a)\right)\to 0$ as $a\to\infty$
    \item $W_0\left(f(a)\right)\in(-1,0)$ for $a\in(\abar,\infty)$
\end{itemize}
\end{theorem}
\begin{proof}
See Appendix \ref{sec:app_cons_func}.    
\end{proof}
The main take-away from Theorem \ref{thm:cons_func} is that the expression for optimal consumption is fully explicit in terms of assets, income, interest rate and the borrowing constraint. This expression involves the Lambert W function, which is defined as the inverse of the equation $ye^y=x$. The inverse is not unique and thus has two branches called $W_0$ and $W_{-1}$. The Lambert W function is useful because it has a number of well-known properties \textemdash see \cite{Mezo2022lambert} for a textbook treatment. 

The first bullet point guarantees that a constrained household consumes the available disposable income $r\abar + y$. The second bullet point guarantees that than an arbitrary wealthy household will behave as an unconstrained one given that the consumption function then boils down to $c_u^\ast(a,y;r) = r a+y+\frac{\rho-r}{\gamma r}$, which is the unconstrained consumption function in this setting\textemdash hence the subscript $u$. This explains why the first branch of the Lambert W function is relevant here as the second branch diverges to $-\infty$ when its argument goes to 0. Finally, the third bullet point guarantees that the presence of the borrowing constraint only acts as a discount on the intertemporal substitution motive. Given the assumption that $\rho>r$, this motive is strictly positive and thus $W_0\left(f(a)\right)\in(-1,0)$ guarantees that constrained consumption will be strictly lower than unconstrained consumption. Intuitively, for someone that is close to the constraint $\abar$, consumption smoothing after, say, an increase in the real rate of interest, is less of an option. Even though the household is relatively impatient, consumption in the short run doesn't go up by that much because it increases the risk of hitting the constraint. 

As alluded to in the introduction, it is well known that consumption functions obtained under the assumptions of CARA utility and no borrowing constraints are linear in assets and can exhibit negative consumption. In this context, the expression derived in \ref{thm:cons_func} will be useful to show that the consumption function with a borrowing constraint doesn't inherit any of these counterfactual properties. 

\subsection{A characterization}

The usual motivation for using utility functions of the CARA class is that they feature no wealth effects and give a consumption function that is \textit{linear} in assets. That property has been used several times in the literature following the seminal contributions of \cite{Kimball1989precautionary} and \cite{Caballero1990consumption}. This is however only true in an environment without borrowing constraints, which is different from the one considered in this paper.\footnote{To the best of my knowledge, there is no treatment of an income fluctuation problem that combines a CARA utility function with a borrowing constraint.} Given the form of the optimal consumption function derived in the last subsection, one can anticipate that it won't be true in the current setting. With this in mind, one would still like the consumption function to display properties that are considered standard in the literature: increasing in permanent income and assets and possibly concave in assets. In order to do so, one has to characterize the Jacobian and Hessian matrices associated with the consumption function \eqref{eq:cstar_W0}. These are detailed in the following proposition.

\begin{prop}
\label{prop:JH}
Assume that the consumption function $c^\ast(a,y;r,\abar)$ is defined as in Theorem \ref{thm:cons_func}. Then it follows that its Jacobian vector $\mathbf{J}_c$ and Hessian matrix $\mathbf{H}_c$ are given by:
\begin{align*}
\mathbf{J}_c &\ = 
\begin{bmatrix}
\frac{r}{1+w(a)}   & 1
\end{bmatrix}\quad,\quad
\mathbf{H}_c \ =
\begin{bmatrix}
\frac{\gamma r^3}{\rho-r}\frac{w(a)}{(1+w(a))^3} & 0\\
0 & 0
\end{bmatrix}
\end{align*}     
where I have defined $w(a)\defeq W_0\left(f(a)\right)$ for convenience. Given the properties of $w(a)$ derived in Theorem \ref{thm:cons_func}, the optimal consumption function is increasing in assets and permanent income as well as concave in assets. 
\end{prop}
\begin{proof}
See Appendix \ref{app:proof_prop_JH}    
\end{proof}
The main take-away from Proposition \ref{prop:JH} is that, while consumption is an affine function of income, it is a concave function of assets. This can be clearly seen in Figure \ref{fig:Consumption_Function}, which plots both the constrained and unconstrained consumption functions in terms of initial assets. One can also read off most of the properties detailed in Theorem \ref{thm:cons_func} from that Figure.

\begin{figure}[ht]
	\centering
	\bigskip
	\caption{Consumption functions with and without borrowing constraint}

	{\small
\includegraphics[width=.75\textwidth]{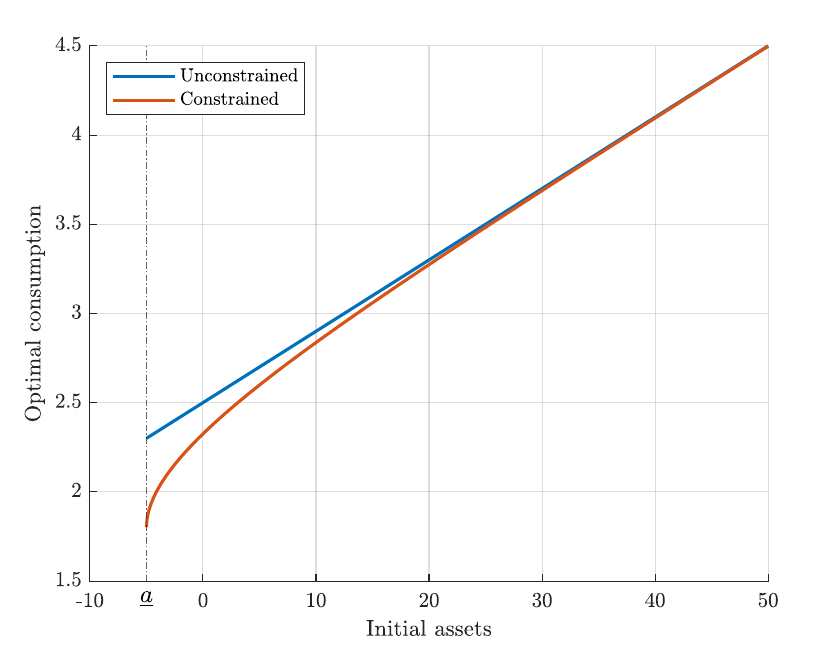}
	}
	\begin{minipage}{0.5\linewidth}
	{
 	 \footnotesize\emph{Notes: For this figure, I use the following parameter values: $r=0.04, \gamma = 2, \rho = 0.08$, $\abar=-5$ and $y=2$.} 
	 }
	\end{minipage}
    \label{fig:Consumption_Function}
\end{figure}

As a result, the consumption function shares some qualitative properties with its CRRA counterpart. With this in mind, I now move on to the main topic of interest, the effect of an increase in the real interest rate.

\section{Consumption and interest rates}
\label{sec:cons_rates}

When does consumption increase after an increase in the interest rate? This is a question that isn't usually covered in the existing literature on income fluctuation problems, which largely focuses on how consumption reacts to \textit{initial assets}. See for example the results in \cite{Benhabib2015wealth} among others. Some contributions derive results about the effects of income in that environment, such as \cite{Ma2020income} and \cite{Roulleau2026explicit} but these are few and far between. There is even less of a focus on the effects of varying the (real) interest rate within this literature. One exception is \cite{Lehrer2018effect}. One obstacle preventing this is the difficulty in getting an explicit mapping from interest rates to consumption in that environment. Armed with the mapping from Theorem \ref{thm:cons_func}, I show in this section that one can get clear-cut results about the effects of interest rates on consumption. 

The focus on interest rates and consumption conversely features front and center in the growing literature on heterogeneous agents New Keynesian (HANK) models. Although virtually all of these models are solved numerically, some do derive a consumption function \textemdash see for example \cite{Bilbiie2020new} and more recently \cite{Bilbiie2025hanksson}. In that case, the focus is decidedly on the interest rate. However, this is done in a simplified environment in which certainty equivalence holds. In addition, the setup in those papers either assumes assets that are in zero net supply for simplicity or considers a log-linear approximation around a steady state with zero assets. That buys a lot of tractability and generates critical insights in a general equilibrium environment, but it makes the focus very different from the literature on income fluctuation problems. 

The goal of this section is to bridge that gap and study the effects of interest rates on consumption in a model which features non-zero assets. In particular, I focus on the response of consumption to an increase in interest rate and how it depends on the household's asset position. The findings reported in \cite{Holm2021transmission} suggest that this dependence is far from trivial. Throughout this section, I will refer to the effect on consumption of a marginal increase in the interest rate as the marginal propensity to consume out of an increase in the interest rate (MPCr). In order to fix ideas and to connect with the existing literature, it will be useful to work first with the version of the model where the borrowing constraint never binds. 

\subsection{An unconstrained benchmark}
\label{sec:benchmark_unconstr}

To obtain the unconstrained benchmark, I consider first the limit as $\abar\to -\infty$. In that case, as we have seen before, the consumption function is given by:
\begin{align}
\label{eq:cstar_unconst}
c_u^*(a,y;r) =  ra+y+\frac{\rho-r}{\gamma r}   
\end{align}
In this context, consumption for a given initial level of assets $a$ is given by two main components. The first one is the annuity value of permanent income with is given by $r$ times $a+y/r$. Since there is no borrowing constraint, the household can borrow at will and look forward far into the future. If the consumption profile was constant, this would be the only component that matters. In general however, the consumption profile will not be constant. Indeed, under CARA preferences, the consumption Euler equation is given by $dc_t/dt=(r-\rho)/\gamma$ which is strictly negative given the impatience condition. The household will consume more in the present and there will be a "shortfall" of consumption in the future. The extra consumption in the present increases by the annuity value of all future shortfalls, the second term in \eqref{eq:cstar_unconst}. 

An interesting knife-edge case arises when $\rho=r$, which is used in \cite{Hall1978stochastic} to obtain the classic result that consumption is a martingale under a quadratic utility function. In that case, the permanent income hypothesis holds and consumption is both $(i)$ constant over time and $(ii)$ equal to its annuity value. This is the case where the inter-temporal substitution motive is absent from the model. This motive features front and center in the recent New Keynesian literature and is the main reason why consumption necessarily declines after an increase in interest rates.\footnote{It should be said that, in the New Keynesian literature, one usually assumes that $\rho=r$ in the long run and then considers a temporary deviation $r+\epsilon$ from the long run interest rate.}

Beyond the transitory versus permanent distinction, the deeper point here is that considering a model that is log-linearized around a steady state with zero assets gives an unambiguous decrease of consumption on impact. Given that the consumption function given in \eqref{eq:cstar_unconst} is clearly non-linear one can see that, under the assumption that $a=0$, consumption necessarily decreases after an increase in $r$. It is then straightforward to notice that, with non-zero initial assets this may not be the case. This is described in the following corollary of Theorem \ref{thm:cons_func}.

\begin{corollary}
\label{cor:threshold_unconstr}
Assume that the impatience condition $\rho>r$ holds and that $\abar\to -\infty$ so that the consumption function is given by \eqref{eq:cstar_unconst}. Then the MPC out of an increase in the interest rate is strictly negative if and only if
\begin{align*}
a<\frac{\rho}{\gamma r^2}   \defeq \overline{a}^u, 
\end{align*}
where the superscript $u$ stands for 'unconstrained'.
\end{corollary}
\begin{proof}
Note that $\partial c_u^*(a, y; r) / \partial r = a-\rho/(\gamma r^2)$, which immediately gives the result.     
\end{proof}
This result arises because of an income effect coming from a positive net asset position. If initial assets are positive, then an increase in the interest rate increases the disposable income of the household. That effect counteracts the inter-temporal substitution effect that goes in the opposite direction. If the initial asset position is such that it is above the threshold derived in Corollary \ref{cor:threshold_unconstr}, then consumption increases after an increase in interest rates. This cannot happen in a standard New Keynesian model because that wealth effect is typically ignored. Note also that if the net asset position is negative, then consumption decreases by even more. In that case, the household is in debt and servicing that debt becomes harder after the increase in interest rates: the household has to cut back on consumption. 

In recent two agents/tractable heterogeneous agents New Keynesian models (see \cite{Bilbiie2025monetary}), the rich agent is usually the one pricing government bonds and whose consumption follows a standard Euler equation. In contrast, the hand to mouth agents consumes their current income every period. In this framework, the consumption of savers decreases directly after an increase in the interest rate and the consumption of hand to mouth agents does too after indirect effects decrease their disposable income. Looking at the evidence presented in \cite{Holm2021transmission} however, reality seems to be a little more complicated. In that paper, the authors show that, at least in the short run, some households increase consumption after an unexpected increase in the interest rate. Moreover, this only concerns households above a threshold in liquid wealth. Corollary \ref{cor:threshold_unconstr} shows that the unconstrained model can replicate that empirical fact.  

This doesn't mean however that the unconstrained model with non-zero initial wealth is the end of the story as it possesses some counterfactual predictions. Given the decreasing path of consumption under the impatience condition $\rho>r$ and the absence of a borrowing constraint, consumption will eventually become negative in finite time.\footnote{Indeed, for a CARA utility function the (lower) \cite{Inada1963two} condition fails given that $u^{'}(0)=1\neq \infty$.} 

This consideration establishes that, while analytically tractable, the unconstrained case may not be a useful benchmark to think about how consumption reacts to changes in the real rate of interest. This provides a motivation to explore the effect of interest rates on consumption in the model where both the borrowing constraint and the non-negativity of consumption are taken seriously. The main point will be to study whether there also exists a threshold of initial wealth above which the MPCr is strictly positive and how that threshold compares with the one obtained in Corollary \ref{cor:threshold_unconstr}. 

\subsection{Consumption and interest rates with a borrowing constraint}

I now move on to a framework where the borrowing constraint is such that $\abar>-y/r$ so that consumption never becomes negative. With this in mind, the natural question that arises is once again: what is the MPCr? Equation \eqref{eq:cstar_W0} suggests that there still exists a threshold of initial assets above which consumption increases after an increase in the interest rate. It turns out that equation \eqref{eq:cstar_W0} lends itself to an intuitive decomposition for the reaction of consumption. I describe this in the following proposition.

\begin{prop}
\label{prop:dcdr_decomp}
Assume that the optimal consumption function is given by equation \eqref{eq:cstar_W0}. Then the MPCr can be written as:
\begin{align}
\label{eq:dcdr_constr}
\frac{\p c^\ast(a,y;r,\abar)}{\p r} = a - \rho\frac{1+w(a)}{\gamma r^2} -\frac{w(a)}{1+w(a)}\frac{2\rho-r}{\rho-r}(a-\abar)> \frac{\p c_u^*(a,y;r)}{\p r},    
\end{align}
where $w(a)\defeq W_0\left(f(a)\right)$ as before and $c_u^*(a,y;r)$ is the unconstrained optimal consumption function given by equation \eqref{eq:cstar_unconst}.
\end{prop}
\begin{proof}
See Appendix \ref{sec:app_dcdr_decomp}. 
\end{proof}
This expression is the explicit, constant-income version of the implicit one presented in the proof of Proposition 5 in \cite{Achdou2022income}, who work with a CRRA utility function with an intertemporal elasticity of substitution restricted to be lower than 1. In their case, consumption necessarily decreases after an increase in interest rates. In my case, this doesn't necessarily happens and the three terms on the right hand side help to understand why. I outline their interpretation in sequence.

The first term on the right hand side of equation \eqref{eq:dcdr_constr} is the same as before: the more assets one has, the higher the income effect coming from a higher interest rate. The second term on the right hand side is also similar to its counterpart without borrowing constraint, except that it is multiplied by $1+w(a)\in (0,1)$. As a result, the inter-temporal substitution effect from before is dampened in the presence of a borrowing constraint. 

The third term on the right hand side is new and arises only when the household has initial assets strictly above the constraint. First, given that $w(a)\in (-1,0)$, note first that this effect is necessary positive. Second, notice that this effect is stronger the further away the household is from the constraint in terms of assets. This term captures the fact that a higher interest rate increases the time to depletion. As a result, the household can enjoy more consumption in the present.

The strict inequality in equation \eqref{eq:dcdr_constr} states that the MPC out of an increase in $r$ is \textit{always} greater with the constraint than without. This doesn't mean that the MPC with a constraint is larger in magnitude, but rather that it is larger in level. As a result, whenever both are negative then the MPC with a constraint will be larger (less negative). Conversely, whenever both are positive, the MPC with the constraint will be larger in magnitude. These two features are represented in Figure \ref{fig:MPCs}, which plots the two MPCs as a function of initial assets.

\begin{figure}[ht]
	\centering
	\bigskip
	\caption{MPCs with and without borrowing constraint}

	{\small
\includegraphics[width=.75\textwidth]{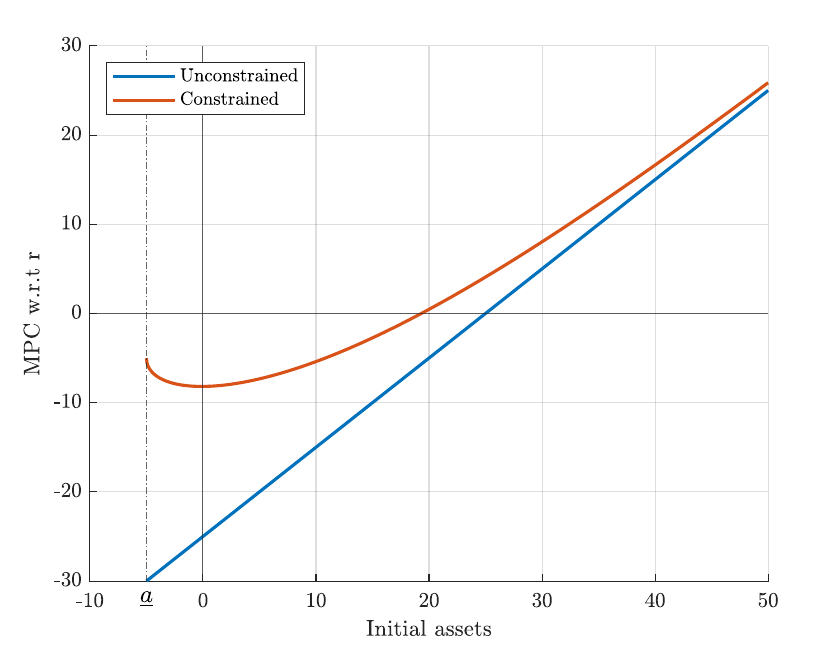}
	}
	\begin{minipage}{0.5\linewidth}
	{
 	 \footnotesize\emph{Notes: For this figure, I use the following parameter values: $r=0.04, \gamma = 2, \rho = 0.08$, $\abar=-5$ and $y=2$.} 
	 }
	\end{minipage}
    \label{fig:MPCs}
\end{figure}

All in all, there are two modifications compared to the unconstrained case: the first one dampens the negative effect on consumption and the second one adds a mechanism that has a positive effect on consumption. As a result, it naturally follows that consumption crowding in after an increase in the interest rate is more likely in that setup. This is illustrated graphically in Figure \ref{fig:MPCs}, where the MPC with the borrowing constraint turns positive for a significantly lower level of initial assets compared to the unconstrained case. I formalize this in the following proposition. 

\begin{prop}
\label{prop:ac_au}
Assume that the optimal consumption function is given by equation \eqref{eq:cstar_W0}. Then the effect of a change in the interest rate on consumption is strictly positive if and only if $a>\overline{a}^c>0$ and is negative below $\overline{a}^c$. Furthermore, the threshold is such that $\overline{a}^c<\overline{a}^u$.
\end{prop}
\begin{proof}
See Appendix \ref{sec:app_ac_au}.   
\end{proof}
Proposition \ref{prop:ac_au} states that there exists a unique threshold of initial assets $\overline{a}^c$ above which the household will increase consumption if the interest rate increases. Furthermore, that threshold is strictly positive so that any initial asset position that is negative gives a decrease in consumption after an increase in $r$. This is consistent with the empirical findings in \cite{Di2017interest}, \cite{Floden2021household} as well as \cite{Cloyne2020monetary} and \cite{Foulis2026interest}.

In addition, the proposition states that this threshold is guaranteed to be \textit{lower} than the one in the case where the household doesn't face a borrowing constraint. It follows that the presence of a borrowing constraint makes it more likely that the MPCr is positive. 

Another notable feature of Figure \ref{fig:MPCs} is that while the MPCr for the unconstrained case is strictly decreasing as initial assets decrease, this is not the case for the constrained MPCr. This is because there are two competing effects at play in the presence of a borrowing constraint. To get a grasp of these two competing effects, first note that the cross-derivative with respect to a change in initial assets and the interest rates doesn't depend on the order in which one takes the partial derivative. This means that one can start from the partial derivative with respect to $a$ from Proposition \ref{prop:JH} and then consider the partial derivative with respect to $r$. 

In a nutshell, one would want to understand why $\p \left[r/(1+w(a))\right]/\p r$
is decreasing at first before turning increasing after, where $r/(1+w(a))$ is the MPC with respect to initial assets derived in Proposition \ref{prop:JH}. When initial assets are close to $\abar$, $w(a)$ is close to $-1$ and thus the MPC with respect to $a$ is rather large: an increase in $r$ will have a relatively large impact by directly increasing income. This effect is proportional to $1/(1+w(a))$ where $w(a)$ is close to $-1$ as $a$ is close to $\abar$. As initial assets increase, $w(a)$ drifts towards 0 and this positive effect becomes less prominent. At the same time, increasing $r$ increases the time to depletion as the initial stock of assets can now be sustained for longer. This positive effect on the MPC with respect to $a$ is all the more stronger the higher the initial assets, \textit{i.e} the greater $a-\abar$. As long as initial assets are such that $a>\tilde{a}$, where $\tilde{a}$ denotes the inflexion point, the second effect dominates and the MPC out of an increase in the interest rate is increasing in initial assets. 

There is an obvious feature that has been a fixture in the literature on heterogeneous agents models and which is missing from the model at the moment: uncertainty. Unless in a rare few cases, it is generally impossible to solve a heterogeneous agent model from the \cite{Bewley1977permanent}-\cite{Huggett1993risk}-\cite{Aiyagari1994uninsured} tradition in closed form.\footnote{See \cite{Toda2017huggett} for a closed form solution modulo a transcendental equation for the equilibrium interest rate in an environment with a CARA utility function but without borrowing constraints. In that case, the fact that consumption is linear in wealth facilitates the aggregation.} This is all the more difficult in the presence of a borrowing constraint as in the current framework. Still, one would like to know whether the presence of, say, income uncertainty increases or decreases the likelihood of the MPCr to turn positive. In addition, one would also like to know whether the result that consumption potentially increases after an increase in the interest rate also happens with a more standard CRRA utility function.

In order to tackle this question, I develop two extensions in the next section. The first one is conceptually close to the one used in \cite{Kimball1990precautionary}.\footnote{Another related paper is \cite{Blanchard1988consumption}, where the authors study a second-order approximation of the Euler equation around the case of no risk. However, they obtain a relationship with two endogenous variables while I obtain results on the equilibrium consumption function.} More precisely, I assume that income $y_t$ follows a two-state Poisson process with values $y-\sigma$ and $y+\sigma$. Crucially, I will derive expressions that are valid in the case where $\sigma$ is close to zero. In other words, I will be working with a model in which risk is arbitrarily low. In the second one, I revert back to constant income and allow for CRRA preferences. 

\section{Extensions}
\label{sec:extensions}
\subsection{Extension $\#$ 1: Uncertainty}
\label{sec:uncertainty}

I now assume that, instead of being constant, income follows a two-state Poisson process with a transition rate given by $\lambda$. That is, income oscillates between two values $y^+=y+\sigma$ and $y^-=y-\sigma$, where $\sigma$ is positive and assumed to be small. Regardless of what state income is currently in, the probability to switch during a small interval $dt$ is given by $\lambda \cdot dt$. This ensures that the income process is symmetric, which is going to help maintaining tractability. Let $s\in\left\{-,+\right\}$ denote the current state of the income process. At each point in time, the value function takes in two arguments: initial assets $a$ and the income state $s$. In order to facilitate notation, I define $V^+(a)\defeq V(a,+)$ and $V^-(a)\defeq V(a,-)$. With this notation in mind, the model boils down to a system of two Hamilton-Jacobi-Bellman (HJB) equations:
\begin{align*}
\rho V^+(a) &= \max_{c\geq 0}\left\{u(c)+(V^+(a))^{'}(ra+y^+-c)+\lambda(V^-(a) - V^+(a))\right\}\\    
\rho V^-(a) &= \max_{c\geq 0}\left\{u(c)+(V^-(a))^{'}(ra+y^--c)+\lambda(V^+(a) - V^-(a))\right\}
\end{align*}
where $(V^s(a))^{'}=\p V(a,s)/\p a$ for $s\in\left\{-,+\right\}$, which is coupled with the borrowing constraint $a\geq \abar$ where $\abar> -y^-/r$. This guarantees that consumption is positive even if the household is stuck at the low state forever. Conditional on being in the high state, the budget constraint is now written as $da/dt= ra_t+y^+-c_t$
and vice versa for the low state. These two HJB equations are derived in the online Appendix. In order to get a closed form expression for the consumption under stochastic income, I look for a solution that is valid for an arbitrarily small $\sigma$. This will require the following assumption:
\begin{Ass}
\label{ass:analytic}
The value functions $V^+(a)$ and $V^-(a)$ are analytic in $\sigma$ in a neighborhood of $\sigma=0$ uniformly in $a$ for $a\in(\abar,\infty)$.  
\end{Ass}
Under Assumption \ref{ass:analytic}, I look for consumption functions that take the form $c^\ast(a,y;r,\abar,+) = c^\ast(a,y;r,\abar) + \sigma\cdot c_1(a,y;r,\abar) + O(\sigma^2)$ for the high income state and $c^\ast(a,y;r,\abar,-) = c^\ast(a,y;r,\abar) - \sigma\cdot c_1(a,y;r,\abar) + O(\sigma^2)$ for the low income state, where $O(\sigma^2)$ denotes term of second order. I show in the online Appendix that these two state-contingent consumption functions can be characterized as follows:
\begin{prop}
\label{prop:cons_func_sto}
Assume that Assumptions \ref{ass:parameters} and  \ref{ass:analytic} hold. Then the state-contingent consumption functions are given by:  
\begin{align}
\label{eq:cons_func_sto_+}
c^\ast(a,y;r,\abar,+) &= ra+y+\frac{\rho-r}{\gamma r}(1+w(a))
+ \sigma c_1(a,y;r,\abar) + O(\sigma^2)\\
c^\ast(a,y;r,\abar,-) &= ra+y+\frac{\rho-r}{\gamma r}(1+w(a)) -\sigma c_1(a,y;r,\abar) + O(\sigma^2),    
\label{eq:cons_func_sto_-}
\end{align}
for all $a\in(\abar,\infty)$ and where I have defined:
\begin{align*}
c_1(a,y;r,\abar) = \frac{r}{1+w(a)}\frac{1-(-w(a))^\nu}{r+2\lambda}, \quad \nu = 1+2\lambda/r
\end{align*}
\end{prop}
\begin{proof}
See online Appendix.    
\end{proof}
Naturally, by construction these two boil down to the consumption function under certain income \eqref{eq:cstar_W0} as $\sigma\to 0$. With this in mind, the goal is to revisit the question that I explored under a fixed income: how likely is it that consumption increases after an increase in the real rate of interest? In particular, the goal is to understand how the presence of stochastic income influences that response. 

Without loss of generality, I focus on the consumption function conditional on being in the high state. In that case, the risk is that income goes down in the near future. Notice that one can rewrite it as $c^\ast(a,y;r,\abar,+) = c_u^\ast(a,y;r)+\sigma c_1(a,y;r,\abar) + O(\sigma^2)$, where the first term on the right hand side is the optimal consumption function under constant income $y$ derived earlier. It follows that the effect of an increase in the interest rate with stochastic income will be given by 
the effect described earlier plus an extra term that has to do with the fact that income is stochastic. More precisely, the effect of an increase in the real interest rate is:
\begin{align}
\label{eq:decomp_dcdr}
\frac{\p c^\ast(a,y;r,\abar,+)}{\p r} = a - \rho\frac{1+w(a)}{\gamma r^2} -\frac{w(a)}{1+w(a)}\frac{2\rho-r}{\rho-r}(a-\abar) + \sigma\frac{\p c_1(a,y;r,\abar)}{\p r}
\end{align}
where I have ignored second order terms. From Proposition \ref{prop:ac_au}, there exists a threshold $\overline{a}^c$ above which the first three terms on the right hand side of equation \eqref{eq:decomp_dcdr} switch from negative to positive. Whether or not the presence of stochastic income increases the likelihood of consumption being crowded in after an increase in the real interest rate crucially hinges on the sign of $\p c_1(a,y;r,\abar)/\p r$. It turns out that the sign of this partial derivative is actually ambiguous as a function of initial assets, just as before. Given the intricate nature of the dependence of this extra term on $r$, it is difficult to prove the existence of a threshold and I give asymptotic results for low and high initial assets in the following proposition:
\begin{prop}
\label{prop:dcdr_sto_+}
Assume that $\rho>r$, that Assumption  \ref{ass:analytic} holds and that $c_1(a,y;r,\abar)$ is defined as in Proposition \ref{prop:cons_func_sto} for all $a\in(\abar,\infty)$. It follows that 
\begin{align}
\label{eq:limit-a_low}
    \frac{\partial }{\partial r}c_1(a,y;r,\abar)
    =
    -\frac{\lambda}{2(\rho-r)}\sqrt{\frac{2\gamma(a-\underline{a})}{\rho-r}}
    + O(a-\underline{a})
    \xrightarrow[a\to\underline{a}^+]{} 0^-.
\end{align}
for initial assets close to the constraint. Further, it also follows that
\begin{align}
\label{eq:limit-a_high}
\frac{\p c_1(a,y;r,\abar)}{\p r }
&\xrightarrow[a\to+\infty]{} \frac{2\lambda}{(r+2\lambda)^2}>0
\end{align}
for high initial assets. 
\end{prop}
\begin{proof}
See online Appendix.    
\end{proof}
The main take-away from Proposition \ref{prop:dcdr_sto_+} is that, in a similar fashion as before with the introduction of a borrowing constraint, the presence of downward income risk adds another channel that increases the possibility of consumption crowding in for high enough assets. Note also that this mechanism plays an opposite role for low enough initial assets. 

In order to understand the intuition behind these results, note that there are two main mechanisms at play: $(i)$ if the household has positive initial assets, an increase in interest rates means that the real value of these precautionary savings has gone up. As a result, for a given probability to fall in the low income state the precautionary motive decreases and the household spends more as a result. At the same time $(ii)$ the present value of human wealth conditional on being in the low income state is $(y-\sigma)/r$. Therefore, an increase in interest rates diminishes the value of human wealth if the household gets a negative income shock. This increases the incentive for precautionary savings and thus to consume less. 

In this context, the main message of Proposition \ref{prop:dcdr_sto_+} is that effect $(i)$ dominates for households with large initial assets. At the same time, effect $(ii)$ dominates for households with low (and potentially negative) initial assets. This is further evidenced by the fact that the MPC out of an increase in the interest rate is increasing linearly with $\lambda$ for low initial assets. This is intuitive: the larger the possibility to fall in the low income state, the stronger the precautionary savings motive and thus the larger the cutback on consumption.

\subsection{Extension $\#$ 2: CRRA Utility Function}
\label{sec:CRRA}

All of the results so far were obtained under the assumption that the utility function has constant absolute risk aversion. While convenient analytically, this specification has a number of shortcomings that make it non standard in income fluctuation problems and in macroeconomics as a whole. Even though the presence of a borrowing constraint alleviates some of these concerns, it still remains that the CARA assumption itself may not be the best description for a typical consumer's behavior. For example, an increase in wealth does not change an agent's absolute willingness to take a fixed dollar risk which may be counterfactual. 

In this context, in the spirit of \cite{Shimer2008liquidity}, I show that the analytically tractable results obtained with the CARA assumption also obtain with a more standard CRRA assumption.\footnote{Note that the CRRA assumption can be extended to the broader HARA family with a simple change of variables. It follows that the results presented in this section also hold for a more general class of utility functions. I zoom in on a CRRA specification as it is arguably the most commonly used.} That being said, the assumption of CRRA will make the income fluctuation problem less tractable which will prevent the study of the role of risky income. As a result, I maintain the assumption of constant income throughout this subsection. The utility function is now given by the usual $u(c) = c^{1-\gamma}/(1-\gamma)$ where $\gamma> 0$ and $u(c)=\ln(c)$ if $\gamma=1$. As before, using the first order conditions of the associated HJB equation with respect to $a$ and $c$ as well as the envelope condition, one can obtain the following differential equation governing the optimal consumption function:
\begin{align}
c(a) = ra + y + \frac{\rho-r}{\gamma }\frac{c(a)}{c'(a)}.
\label{eq:ODE_CRRA}
\end{align}
In sharp contrast with equation \eqref{eq:ODE}, this non-linear ODE cannot be solved explicitly. To circumvent this difficulty, one can follow \cite{Holm2018consumption} and derive an implicit equation for $c(a)$. This equation is described in the following proposition.

\begin{prop}
\label{prop:F_CRRA}
Assume that Assumption \ref{ass:parameters} holds. I define $c_{min}\defeq r\abar+y$, which is the minimum consumption level which is attained whenever $a=\abar$. Then the optimal consumption function $c(a)$ is the solution to the following equation: 
\begin{align}
F(c(a),a;r) = c(a) - (1+\alpha\theta)(ra+y)+\alpha\theta c_{min}\left(\frac{c_{min}}{c(a)}\right)^\kappa = 0
\end{align}
where $\alpha=1/\gamma$, $\theta=(\rho-r)/r$ and 
$\kappa=1/(\alpha\theta)=\gamma r/(\rho-r)$.
\end{prop}
\begin{proof}
See Appendix \ref{sec:app_F_CRRA}.  
\end{proof}

The main intuition for the approach is that implicit differentiation allows one to write that $\p c(a;r)/\p r = -F_r/F_c$
where $F_c$ denotes the partial derivative of $F$ with respect to $c(a;r)$ and, likewise, $F_r$ denotes the partial derivative of $F$ with respect to $r$. It turns out that $F_c$ will have a non-ambiguous sign so that whether or not consumption is crowded in after an increase in the interest rate only depends on the sign of $F_r$. The conditions for the crowding in are explained in the following theorem:
\begin{theorem}
\label{thm:dcdr_CRRA}
Assume that Assumption \ref{ass:parameters} holds. It follows that the optimal consumption function with a CRRA utility function is such that:
\begin{enumerate}
    \item $\p c^*(a;r)/\p r\leq 0$ if $\gamma\leq 1$
    \item If $\gamma>1$, then there exists $\overline{a}^{CRRA}\in (\abar,\infty)$ such that $\p c^*(a;r)/\p r> 0$ if and only if $a>\overline{a}^{CRRA}$.
\end{enumerate}
\end{theorem}
\begin{proof}
See Appendix \ref{sec:app_CRRA}.    
\end{proof}
The main take-away from Theorem \ref{thm:dcdr_CRRA} is that if the coefficient of relative risk aversion is less than 1 (equivalently, if the elasticity of intertemporal substitution (EIS) is above 1), then consumption is necessarily crowded out after an increase in the interest rate. This is because substitution effects dominate in this case as in \cite{Lehrer2018effect}. One implication of this result is that an income fluctuation problem with a CRRA utility function with EIS above 1 cannot replicate the evidence reported in \cite{Holm2021transmission}.

If $\gamma>1$ instead, substitution effects do not dominate and there is a scope for wealth effects to ensure that consumption is crowded in. As before, this happens when initial wealth crosses a certain threshold. This situation corresponds to an EIS that is strictly below 1 which seems to be the empirically relevant case\textemdash see \cite{Havranek2015measuring}. It follows that the result of a crowding in of consumption after an increase in interest rates for sufficiently high wealth isn't specific to the CARA assumption. 

\section{Conclusion}
\label{sec:conclusion}

In a recent contribution, \cite{Achdou2022income} put forward the benefits of continuous-time formulations of standard heterogeneous agents models. In particular, they show that in some special cases these are more amenable to closed-form solutions that yield valuable insights. In this context, the current paper makes the case that the assumption of a CARA utility function further improves tractability and doesn't have to come with undesirable properties such as negative consumption. In addition, I show that the Lambert W function is a useful tool to have as a theoretical macroeconomist. 

Leveraging these two features, I have derived an explicit mapping from initial assets, permanent income as well as real interest rates to optimal consumption in the presence of borrowing constraints. In turn, I have shown that there exists a threshold level of initial assets above which consumption increases after an increase in the interest rate because of wealth effects. Finally, I have shown that the result of a crowding in is not specific to the CARA assumption as it also happens with a more standard CRRA specification.

\bibliographystyle{apalike2}
\bibliography{CARA_dcdr}

@article{Acharya2020understanding,
  title={Understanding HANK: Insights from a PRANK},
  author={Acharya, Sushant and Dogra, Keshav},
  journal={Econometrica},
  volume={88},
  number={3},
  pages={1113--1158},
  year={2020},
  publisher={Wiley Online Library}
}

@article{Acharya2023optimal,
  title={Optimal monetary policy according to HANK},
  author={Acharya, Sushant and Challe, Edouard and Dogra, Keshav},
  journal={American Economic Review},
  volume={113},
  number={7},
  pages={1741--1782},
  year={2023},
  publisher={American Economic Association 2014 Broadway, Suite 305, Nashville, TN 37203}
}

@article{Achdou2022income,
  title={Income and wealth distribution in macroeconomics: A continuous-time approach},
  author={Achdou, Yves and Han, Jiequn and Lasry, Jean-Michel and Lions, Pierre-Louis and Moll, Benjamin},
  journal={The Review of Economic Studies},
  volume={89},
  number={1},
  pages={45--86},
  year={2022},
  publisher={Oxford University Press}
}

@article{Aiyagari1994uninsured,
  title={Uninsured idiosyncratic risk and aggregate saving},
  author={Aiyagari, S Rao},
  journal={The Quarterly Journal of Economics},
  volume={109},
  number={3},
  pages={659--684},
  year={1994},
  publisher={MIT Press}
}

@article{Auclert2019monetary,
  title={Monetary policy and the redistribution channel},
  author={Auclert, Adrien},
  journal={American Economic Review},
  volume={109},
  number={6},
  pages={2333--2367},
  year={2019},
  publisher={American Economic Association 2014 Broadway, Suite 305, Nashville, TN 37203}
}

@article{Bauer2023alternative,
  title={An alternative explanation for the "fed information effect"},
  author={Bauer, Michael D and Swanson, Eric T},
  journal={American Economic Review},
  volume={113},
  number={3},
  pages={664--700},
  year={2023},
  publisher={American Economic Association 2014 Broadway, Suite 305, Nashville, TN 37203}
}

@article{Benhabib2015wealth,
  title={The wealth distribution in Bewley economies with capital income risk},
  author={Benhabib, Jess and Bisin, Alberto and Zhu, Shenghao},
  journal={Journal of Economic Theory},
  volume={159},
  pages={489--515},
  year={2015},
  publisher={Elsevier}
}

@article{Bewley1977permanent,
  title={The permanent income hypothesis: A theoretical formulation},
  author={Bewley, Truman},
  journal={Journal of Economic Theory},
  volume={16},
  number={2},
  pages={252--292},
  year={1977},
  publisher={Elsevier}
}

@article{Bilbiie2008limited,
  title={Limited asset markets participation, monetary policy and (inverted) aggregate demand logic},
  author={Bilbiie, Florin O},
  journal={Journal of economic theory},
  volume={140},
  number={1},
  pages={162--196},
  year={2008},
  publisher={Elsevier}
}

@article{Bilbiie2020new,
  title={The new keynesian cross},
  author={Bilbiie, Florin O},
  journal={Journal of Monetary Economics},
  volume={114},
  pages={90--108},
  year={2020},
  publisher={Elsevier}
}

@article{Bilbiie2025monetary,
  title={Monetary policy and heterogeneity: An analytical framework},
  author={Bilbiie, Florin O},
  journal={Review of Economic Studies},
  volume={92},
  number={4},
  pages={2398--2436},
  year={2025},
  publisher={Oxford University Press UK}
}

@article{Bilbiie2025hanksson,
  title={Hanksson},
  author={Bilbiie, Florin O and Galaasen, Sigurd M{\o}lster and G{\"u}rkayna, RS and M{\ae}hlum, Mathis and Molnar, Krisztina},
  year={2025},
  publisher={Faculty of Economics, University of Cambridge}
}

@misc{Blanchard1988consumption,
  title={Consumption: Beyond certainty equivalence},
  author={Blanchard, Olivier J and Mankiw, N Gregory},
  year={1988},
  publisher={National Bureau of Economic Research Cambridge, Mass., USA}
}

@article{Caballero1990consumption,
  title={Consumption puzzles and precautionary savings},
  author={Caballero, Ricardo J},
  journal={Journal of Monetary Economics},
  volume={25},
  number={1},
  pages={113--136},
  year={1990},
  publisher={Elsevier}
}

@article{Carroll1996concavity,
  title={On the concavity of the consumption function},
  author={Carroll, Christopher D and Kimball, Miles S},
  journal={Econometrica: Journal of the Econometric Society},
  pages={981--992},
  year={1996},
  publisher={JSTOR}
}

@article{Clarida1999science,
  title={The science of monetary policy: a new Keynesian perspective},
  author={Clarida, Richard and Gali, Jordi and Gertler, Mark},
  journal={Journal of economic literature},
  volume={37},
  number={4},
  pages={1661--1707},
  year={1999},
  publisher={American Economic Association}
}

@article{Cloyne2020monetary,
  title={Monetary policy when households have debt: new evidence on the transmission mechanism},
  author={Cloyne, James and Ferreira, Clodomiro and Surico, Paolo},
  journal={The Review of Economic Studies},
  volume={87},
  number={1},
  pages={102--129},
  year={2020},
  publisher={Oxford University Press}
}

@article{Coibion2017innocent,
  title={Innocent Bystanders? Monetary policy and inequality},
  author={Coibion, Olivier and Gorodnichenko, Yuriy and Kueng, Lorenz and Silvia, John},
  journal={Journal of Monetary Economics},
  volume={88},
  pages={70--89},
  year={2017},
  publisher={Elsevier}
}

@techreport{Commault2025Permanent_income,
  title={Heterogeneity in MPC Beyond Liquidity Constraints: The Role of Permanent Earnings},
  author={Commault, Jeanne},
  institution={Sciences Po Economics Discussion Paper, hal-03870685},
  year={2025}
}

@misc{Deaton1989saving,
  title={Saving and liquidity constraints},
  author={Deaton, Angus},
  year={1989},
  publisher={national bureau of economic research Cambridge, Mass., USA}
}

@article{Debortoli2024idiosyncratic,
  title={Idiosyncratic income risk and aggregate fluctuations},
  author={Debortoli, Davide and Gal{\'\i}, Jordi},
  journal={American Economic Journal: Macroeconomics},
  volume={16},
  number={4},
  pages={279--310},
  year={2024},
  publisher={American Economic Association 2014 Broadway, Suite 305, Nashville, TN 37203-2425}
}

@article{Di2017interest,
  title={Interest rate pass-through: Mortgage rates, household consumption, and voluntary deleveraging},
  author={Di Maggio, Marco and Kermani, Amir and Keys, Benjamin J and Piskorski, Tomasz and Ramcharan, Rodney and Seru, Amit and Yao, Vincent},
  journal={American Economic Review},
  volume={107},
  number={11},
  pages={3550--3588},
  year={2017},
  publisher={American Economic Association 2014 Broadway, Suite 305, Nashville, TN 37203}
}

@article{Doepke2006inflation,
  title={Inflation and the redistribution of nominal wealth},
  author={Doepke, Matthias and Schneider, Martin},
  journal={Journal of Political Economy},
  volume={114},
  number={6},
  pages={1069--1097},
  year={2006},
  publisher={The University of Chicago Press}
}

@article{Eichenbaum2022state,
  title={State-dependent effects of monetary policy: The refinancing channel},
  author={Eichenbaum, Martin and Rebelo, Sergio and Wong, Arlene},
  journal={American Economic Review},
  volume={112},
  number={3},
  pages={721--761},
  year={2022},
  publisher={American Economic Association 2014 Broadway, Suite 305, Nashville, TN 37203}
}

@techreport{Farhi2022price,
  title={Price theory for incomplete markets},
  author={Farhi, Emmanuel and Olivi, Alan and Werning, Iv{\'a}n},
  year={2022},
  institution={National Bureau of Economic Research}
}

@book{Fisher1930theory,
  title={The Theory of Interest: As Determined by Impatience to Spend Income and Opportunity to Invest It},
  author={Fisher, Irving},
  year={1930},
  publisher={Macmillan},
  address={New York}
}

@article{Fischer2026concave,
  title={Concave consumption functions-A closed-form characterization},
  author={Fischer, Thomas},
  journal={Available at SSRN 4735352},
  year={2026}
}

@article{Floden2021household,
  title={Household debt and monetary policy: Revealing the cash-flow channel},
  author={Flod{\'e}n, Martin and Kilstr{\"o}m, Matilda and Sigurdsson, J{\'o}sef and Vestman, Roine},
  journal={The Economic Journal},
  volume={131},
  number={636},
  pages={1742--1771},
  year={2021},
  publisher={Oxford University Press}
}

@techreport{Foulis2026interest,
  title={How Do Interest Rates Affect Consumption? Household Debt and the Role of Asset Prices},
  author={Foulis, Angus K and Hazell, Jonathon and Mian, Atif R and Tracey, Belinda},
  year={2026},
  institution={National Bureau of Economic Research}
}

@article{Gertler2015monetary,
  title={Monetary policy surprises, credit costs, and economic activity},
  author={Gertler, Mark and Karadi, Peter},
  journal={American Economic Journal: Macroeconomics},
  volume={7},
  number={1},
  pages={44--76},
  year={2015},
  publisher={American Economic Association 2014 Broadway, Suite 305, Nashville, TN 37203-2425}
}

@article{Gouin2023pareto,
  title={Pareto extrapolation: An analytical framework for studying tail inequality},
  author={Gouin-Bonenfant, {\'E}milien and Toda, Alexis Akira},
  journal={Quantitative Economics},
  volume={14},
  number={1},
  pages={201--233},
  year={2023},
  publisher={Wiley Online Library}
}

@article{Hall1978stochastic,
  title={Stochastic implications of the life cycle-permanent income hypothesis: theory and evidence},
  author={Hall, Robert E},
  journal={Journal of Political Economy},
  volume={86},
  number={6},
  pages={971--987},
  year={1978},
  publisher={The University of Chicago Press}
}

@article{Havranek2015measuring,
  title={Measuring intertemporal substitution: The importance of method choices and selective reporting},
  author={Havr{\'a}nek, Tom{\'a}{\v{s}}},
  journal={Journal of the European Economic Association},
  volume={13},
  number={6},
  pages={1180--1204},
  year={2015},
  publisher={Oxford University Press}
}

@article{Helpman1981optimal,
  title={Optimal spending and money holdings in the presence of liquidity constraints},
  author={Helpman, Elhanan},
  journal={Econometrica: Journal of the Econometric Society},
  pages={1559--1570},
  year={1981},
  publisher={JSTOR}
}

@article{Holm2018consumption,
  title={Consumption with liquidity constraints: An analytical characterization},
  author={Holm, Martin Blomhoff},
  journal={Economics Letters},
  volume={167},
  pages={40--42},
  year={2018},
  publisher={Elsevier}
}

@article{Holm2021transmission,
  title={The transmission of monetary policy under the microscope},
  author={Holm, Martin Blomhoff and Paul, Pascal and Tischbirek, Andreas},
  journal={Journal of Political Economy},
  volume={129},
  number={10},
  pages={2861--2904},
  year={2021},
  publisher={The University of Chicago Press Chicago, IL}
}

@article{Huggett1993risk,
  title={The risk-free rate in heterogeneous-agent incomplete-insurance economies},
  author={Huggett, Mark},
  journal={Journal of Economic Dynamics and Control},
  volume={17},
  number={5-6},
  pages={953--969},
  year={1993},
  publisher={Elsevier}
}

@article{Inada1963two,
  title={On a two-sector model of economic growth: Comments and a generalization},
  author={Inada, Ken-Ichi},
  journal={The Review of Economic Studies},
  volume={30},
  number={2},
  pages={119--127},
  year={1963},
  publisher={Wiley-Blackwell}
}

@article{Kaplan2018monetary,
  title={Monetary policy according to HANK},
  author={Kaplan, Greg and Moll, Benjamin and Violante, Giovanni L},
  journal={American Economic Review},
  volume={108},
  number={3},
  pages={697--743},
  year={2018},
  publisher={American Economic Association 2014 Broadway, Suite 305, Nashville, TN 37203}
}

@book{Kase2025estimating,
  title={Estimating nonlinear heterogeneous agents models with neural networks},
  author={Kase, Hanno and Melosi, Leonardo and Rottner, Matthias},
  year={2025},
  publisher={Centre for Economic Policy Research}
}

@book{Keynes1936general,
  title={The General Theory of Employment, Interest, and Money},
  author={Keynes, John Maynard},
  year={1936},
  publisher={Macmillan},
  address={London}
}

@article{Kimball1990precautionary,
  title={Precautionary Saving in the Small and in the Large},
  author={Kimball, Miles S},
  journal={Econometrica: Journal of the Econometric Society},
  pages={53--73},
  year={1990},
  publisher={JSTOR}
}

@article{Kimball1989precautionary,
  title={Precautionary saving and the timing of taxes},
  author={Kimball, Miles S and Mankiw, N Gregory},
  journal={Journal of Political Economy},
  volume={97},
  number={4},
  pages={863--879},
  year={1989},
  publisher={The University of Chicago Press}
}

@article{Lee2025wealth,
  title={Wealth inequality, aggregate consumption, and macroeconomic trends under incomplete markets},
  author={Lee, Byoungchan},
  journal={American Economic Journal: Macroeconomics},
  volume={17},
  number={1},
  pages={414--442},
  year={2025},
  publisher={American Economic Association 2014 Broadway, Suite 305, Nashville, TN 37203-2425}
}

@article{Lehrer2018effect,
  title={The effect of interest rates on consumption in an income fluctuation problem},
  author={Lehrer, Ehud and Light, Bar},
  journal={Journal of Economic Dynamics and Control},
  volume={94},
  pages={63--71},
  year={2018},
  publisher={Elsevier}
}

@article{Ma2020income,
  title={The income fluctuation problem and the evolution of wealth},
  author={Ma, Qingyin and Stachurski, John and Toda, Alexis Akira},
  journal={Journal of Economic Theory},
  volume={187},
  pages={105003},
  year={2020},
  publisher={Elsevier}
}

@article{Ma2021theory,
  title={A theory of the saving rate of the rich},
  author={Ma, Qingyin and Toda, Alexis Akira},
  journal={Journal of Economic Theory},
  volume={192},
  pages={105193},
  year={2021},
  publisher={Elsevier}
}

@article{Ma2022asymptotic,
  title={Asymptotic linearity of consumption functions and computational efficiency},
  author={Ma, Qingyin and Toda, Alexis Akira},
  journal={Journal of Mathematical Economics},
  volume={98},
  pages={102562},
  year={2022},
  publisher={Elsevier}
}

@article{Ma2026optimal,
  title={Optimal Savings under Transition Uncertainty and Learning Dynamics},
  author={Ma, Qingyin and Zhang, Xinxin},
  journal={arXiv preprint arXiv:2603.08663},
  year={2026}
}

@article{Ma2025theory,
  title={A Theory of Saving under Risk Preference Dynamics},
  author={Ma, Qingyin and Song, Xinxi and Toda, Alexis Akira},
  journal={arXiv preprint arXiv:2511.03142},
  year={2025}
}

@article{Merton1971optimum,
  title={Optimum consumption and portfolio rules in a continuous-time model},
  author={Merton, Robert},
  journal={Journal of Economic Theory},
  volume={3},
  number={4},
  pages={373--413},
  year={1971},
  publisher={Elsevier}
}

@book{Mezo2022lambert,
  title={The Lambert W function: its generalizations and applications},
  author={Mezo, Istvan},
  year={2022},
  publisher={Chapman and Hall/CRC}
}

@article{Miranda2021transmission,
  title={The transmission of monetary policy shocks},
  author={Miranda-Agrippino, Silvia and Ricco, Giovanni},
  journal={American Economic Journal: Macroeconomics},
  volume={13},
  number={3},
  pages={74--107},
  year={2021},
  publisher={American Economic Association 2014 Broadway, Suite 305, Nashville, TN 37203-2425}
}

@article{Nakamura2018high,
  title={High-frequency identification of monetary non-neutrality: the information effect},
  author={Nakamura, Emi and Steinsson, J{\'o}n},
  journal={The Quarterly Journal of Economics},
  volume={133},
  number={3},
  pages={1283--1330},
  year={2018},
  publisher={Oxford University Press}
}

@article{Park2006analytical,
  title={An analytical solution to the inverse consumption function with liquidity constraints},
  author={Park, Myung-Ho},
  journal={Economics Letters},
  volume={92},
  number={3},
  pages={389--394},
  year={2006},
  publisher={Elsevier}
}

@article{Ramsey1928mathematical,
  title={A mathematical theory of saving},
  author={Ramsey, Frank Plumpton},
  journal={The Economic Journal},
  volume={38},
  number={152},
  pages={543--559},
  year={1928},
  publisher={Oxford University Press Oxford, UK}
}

@article{Romer2004new,
  title={A new measure of monetary shocks: Derivation and implications},
  author={Romer, Christina D and Romer, David H},
  journal={American economic review},
  volume={94},
  number={4},
  pages={1055--1084},
  year={2004},
  publisher={American Economic Association}
}

@article{Roulleau2026explicit,
  title={Explicit consumption functions with borrowing constraints: A continuous-time approach},
  author={Roulleau-Pasdeloup, Jordan},
  journal={Journal of Monetary Economics},
  pages={103919},
  year={2026},
  publisher={Elsevier}
}

@article{Schechtman1976income,
  title={An income fluctuation problem},
  author={Schechtman, Jack},
  journal={Journal of Economic Theory},
  volume={12},
  number={2},
  pages={218--241},
  year={1976},
  publisher={Elsevier}
}

@article{Schechtman1977some,
  title={Some results on "an income fluctuation problem"},
  author={Schechtman, Jack and Escudero, Vera LS},
  journal={Journal of Economic Theory},
  volume={16},
  number={2},
  pages={151--166},
  year={1977},
  publisher={Elsevier}
}

@article{Shimer2008liquidity,
  title={Liquidity and Insurance for the Unemployed},
  author={Shimer, Robert and Werning, Iv{\'a}n},
  journal={American Economic Review},
  volume={98},
  number={5},
  pages={1922--1942},
  year={2008},
  publisher={American Economic Association}
}

@article{Toda2017huggett,
  title={Huggett economies with multiple stationary equilibria},
  author={Toda, Alexis Akira},
  journal={Journal of Economic Dynamics and Control},
  volume={84},
  pages={77--90},
  year={2017},
  publisher={Elsevier}
}

@article{Toda2021necessity,
  title={Necessity of hyperbolic absolute risk aversion for the concavity of consumption functions},
  author={Toda, Alexis Akira},
  journal={Journal of Mathematical Economics},
  volume={94},
  pages={102460},
  year={2021},
  publisher={Elsevier}
}

@article{Toda2026characterization,
  title={Characterization of Concave Consumption Functions under Conditional Impatience},
  author={Toda, Alexis Akira},
  journal={arXiv preprint arXiv:2608.29488},
  year={2026}
}

@article{Wang2003caballero,
  title={Caballero meets Bewley: The permanent-income hypothesis in general equilibrium},
  author={Wang, Neng},
  journal={American Economic Review},
  volume={93},
  number={3},
  pages={927--936},
  year={2003},
  publisher={American Economic Association}
}

@article{Wang2007equilibrium,
  title={An equilibrium model of wealth distribution},
  author={Wang, Neng},
  journal={Journal of Monetary Economics},
  volume={54},
  number={7},
  pages={1882--1904},
  year={2007},
  publisher={Elsevier}
}

@article{Zeldes1989optimal,
  title={Optimal consumption with stochastic income: Deviations from certainty equivalence},
  author={Zeldes, Stephen P},
  journal={The Quarterly Journal of Economics},
  volume={104},
  number={2},
  pages={275--298},
  year={1989},
  publisher={MIT Press}
}

\appendix

\section{Proof of Theorem \ref{thm:cons_func}}
\label{sec:app_cons_func}

Let us start from equation \eqref{eq:ODE}, which I reproduce here for convenience:
\begin{align*}
c(a) = ra+y + \frac{\rho-r}{\gamma c'(a)};
\end{align*}
Now define the optimal savings function as $s(a) = ra+y-c(a)$. It follow immediately that $s'(a) = r+y-c'(a)$. Therefore, one can rewrite the implicit equation \eqref{eq:ODE} for the consumption function in terms of the savings function as 
\begin{align*}
s'(a) = r + \frac{r-\rho}{\gamma s(a)}.
\end{align*}
Using the fact that $s'(a) = ds/da$, now write:
\begin{align*}
\frac{\gamma s(a)}{r\gamma s(a)+\rho-r}ds = da    
\end{align*}
Now integrating from $\abar$ to a generic $a$ on the right hand side and the corresponding $s(\abar)=0$ to $s(a)$ on the left hand side, one obtains:
\begin{align}
\int_{0}^{s(a)}\frac{\gamma s(a)}{r\gamma s(a)+\rho-r}ds = a-\abar    
\label{eq:integrand}
\end{align}
Notice that one can write the integrand in terms of partial fractions as follows:
\begin{align*}
\frac{\gamma s(a)}{r\gamma s(a)+\rho-r} &= \frac{1}{r}\frac{\gamma r s(a) + \rho-r+r-\rho}{\gamma r s(a) + \rho-r}\\
&=\frac{1}{r}\left[1-\frac{\rho-r}{\gamma r s(a) + \rho-r}\right] 
\end{align*}
Therefore, the left hand side of equation \eqref{eq:integrand} evaluates to:
\begin{align*}
\frac{s(a)}{r} - \frac{\rho-r}{\gamma r^2}\ln\left\{\frac{\gamma r s(a)+\rho-r}{\rho-r}\right\},    
\end{align*}
where the argument of the $\ln$ in the second term is guaranteed to be strictly positive since the unconstrained savings function is given by $s_u(a)=(r-\rho)/(\gamma r)<s(a)$. This in turn implies that $\gamma r s(a)+\rho-r>0$ for all $(\abar,\infty)$. Putting it all together and defining $\tilde{s}(a)\defeq \gamma r s(a)/(\rho-r)$ for convenience, equation \eqref{eq:integrand} becomes:
\begin{align*}
\frac{1}{r}\frac{\rho-r}{\gamma r^2}\tilde{s}(a) - \frac{\rho-r}{\gamma r^2}\ln\left\{1+\tilde{s}(a)\right\} &= a - \abar\\
\Rightarrow\quad  \tilde{s}(a) - \ln\left\{1+\tilde{s}(a)\right\} &= \frac{\gamma r^2}{\rho-r}(a - \abar).
\end{align*}
Now define $\tilde{w}(a)\defeq 1 + \tilde{s}(a)$. The equation then becomes:
\begin{align*}
\tilde{w}(a)-1-\ln\left\{\tilde{w}(a)\right\} &= \frac{\gamma r^2}{\rho-r}(a - \abar)\\
\Rightarrow\quad \ln\left\{\tilde{w}(a)\right\} - \tilde{w}(a) &= -\left[1+\frac{\gamma r^2}{\rho-r}(a - \abar)\right].
\end{align*}
Now taking the exponential on both sides of that equation and then multiplying by -1, one gets:
\begin{align*}
- \tilde{w}(a)e^{-\tilde{w}(a)} = -e^{-\left[1+\frac{\gamma r^2}{\rho-r}(a - \abar)\right]}.   
\end{align*}
Now setting $w(a)=-\tilde{w}(a)$, notice that the solution of this equation is the Lambert W function evaluated at the right hand side by definition. Therefore, one can write
\begin{align}
w(a) = W\left(-e^{-\left[1+\frac{\gamma r^2}{\rho-r}(a - \abar)\right]}\right). 
\label{eq:solW}
\end{align}
Reverting the change of variables, one can write that $w(a) = -1-\tilde{s}(a) = w(a)$. From the definition of $\tilde{s}(a)$, one then gets the optimal savings function
\begin{align*}
s(a) = -\frac{\rho-r}{\gamma r}(1+w(a)).   
\end{align*}
Plugging this into the flow budget constraint, one finally gets:
\begin{align*}
c(a) &= ra+y-s(a)\\
&= ra+y+\frac{\rho-r}{\gamma r}(1+w(a)).
\end{align*}
In the limit as $a\to\infty$ (or, equivalently, as $\abar\to -\infty)$, this consumption function must boil down to the one for an unconstrained agent. In turn, this implies that $w(a)\to 0$ as $a\to\infty$. As $a\to\infty$, the argument of W in equation \eqref{eq:solW} goes to 0. This is where the two branches of the Lambert W function differ. While $W_0(x)\to 0$ as $x\to 0$, $W_0(x)\to -\infty$ in the same context. This is why the first branch of the Lambert W function is the relevant branch here. This establishes the second bullet point of Theorem \ref{thm:cons_func}. 

Now in the limit as $a\to \abar$, the argument in the Lambert W function goes to $-1/e$. By the properties of the Lambert W function, $W_0(-1/e)=-1$. This establishes the first bullet point of Theorem \ref{thm:cons_func}. Note in passing that this case doesn't help in discriminating between the two branches because the two branches give the same value at that point: $W_0(-1/e)=W_{-1}(-1/e)=-1$. Finally, the first branch of the Lambert W function is such that, for $x\in (-1/e,0)$ one has $W_0(x)\in (-1,0)$. This establishes the third bullet point of Theorem \ref{thm:cons_func} and completes the proof.

\section{Proof of Proposition \ref{prop:JH}}
\label{app:proof_prop_JH}
The fact that the first entry is equal to 1 is straightforward from the definition of $c^\ast(a,y;r,\abar)$. It naturally follows that the optimal consumption function is strictly increasing in permanent income. The second entry is given by
\begin{align}
\label{eq:app_Hessian_dcsta_a}
\frac{\p c^\ast(a,y;r,\abar)}{\p a} = -\frac{\rho-r}{\gamma r}\frac{1}{w(a)}\frac{\p w(a)}{\p a}, \quad w\defeq W_0\left(-e^{-\left(1+(a-\abar)\frac{\gamma r^2}{\rho-r}\right)}\right)   
\end{align}
Now define $f(a)=-e^{-\left(1+(a-\abar)\frac{\gamma r^2}{\rho-r}\right)}$ as in Theorem \ref{thm:cons_func}. Note that the Lambert W is such that:
\begin{align*}
\frac{\p w}{\p f(a)} = \frac{w}{f(a)(1+w)}    
\end{align*}
It follows that, using the chain rule, I obtain:
\begin{align*}
\frac{\p w}{\p a} &= \frac{w}{f(a)(1+w)}\frac{\p f(a)}{\p a}     \\
&= \frac{w}{f(a)(1+w)}\left(-\frac{\gamma r^2}{\rho-r}f(a)\right)\\
&= -\frac{w}{1+w}\frac{\gamma r^2}{\rho-r},
\end{align*}
where I have used the properties of the exponential. Plugging this into \eqref{eq:app_Hessian_dcsta_a}:
\begin{align}
\notag
\frac{\p c^\ast(a,y;r,\abar)}{\p a} &= r-\frac{\rho-r}{\gamma r}\frac{w(a)}{1+w(a)}\frac{\gamma r^2}{\rho-r}\\
&= r-\frac{rw(a)}{1+w(a)}\\
&= \frac{r}{1+w(a)},
\label{eq:app_dcda_star}
\end{align}
which completes the derivation of $\mathbf{J}_c$. The fact that the optimal consumption function is increasing in initial assets comes from the fact that the first branch of the Lambert W function is such that $-1<w(a)<0$ for any $f(a)\in(-1/e,0)$.

Moving on to the Hessian matrix, the fact that the first partial derivative with respect to $y$ is 0 implies that the second as well as the cross derivatives are also 0. The only entry that one needs to compute is the second-order derivative with respect to initial assets $a$. Starting from equation \eqref{eq:app_dcda_star} and using the chain rule:
\begin{align*}
\frac{\p^2 c^\ast(a,y;r,\abar)}{\p a^2} &= -\frac{r}{(1+w(a))^2}\frac{w(a)}{f(a)(1+w(a))}\frac{\p f(a)}{\p a}\\
&=\frac{r}{(1+w(a))^2}\frac{w(a)}{f(a)(1+w(a))}\frac{\gamma r^2}{\rho-r} f(a)\\
&= \frac{\gamma r^3}{\rho-r}\frac{w(a)}{(1+w(a))^3}<0
\end{align*}
where the sign on the last line follows from the fact that it has been established that $-1<w(a)<0$. It follows that $1+w(a)>0$ and thus $w(a)/(1+w(a))<0$ for all $a,y\geq 0$. As a result, the consumption function is concave in assets.

\section{Proof of Proposition \ref{prop:dcdr_decomp}}
\label{sec:app_dcdr_decomp}
Given the consumption function from Theorem \ref{thm:cons_func} and taking the partial derivative with respect to $r$, one gets:
\begin{align*}
\frac{\p c^\ast(a,y;r,\abar)}{\p r} = a - \frac{\rho}{\gamma r^2}(1+w(a)) +\frac{\rho-r}{\gamma r}\frac{\p w(a)}{\p r}    
\end{align*}
As before, $w(a)\defeq W_0(f(a))$ where $f(a)=-e^{-z(a)}$ and where $z(a)\defeq \left(1+(a-\abar)\frac{\gamma r^2}{\rho-r}\right)$. It follows from the chain rule that $\p w(a)/\p r = (\p w(a)/\p f(a))\cdot(\p f(a)/\p z(a))\cdot(\p z(a)/\p r)$. These can be written as follows:
\begin{align*}
\frac{\p w}{\p f} = \frac{w(a)}{f(1+w(a))}\  \frac{\p f}{\p z} = e^{-z}=-f\  
\frac{\p z}{\p r} = (a-\abar)\gamma r\frac{2\rho-r}{(\rho-r)^2}
\end{align*}
Putting it all together, I can finally write:
\begin{align*}
\frac{\p w(a)}{\p r} = -\frac{w(a)}{1+w(a)}\frac{2\rho-r}{(\rho-r)^2}(a-\abar)\gamma r    
\end{align*}
Plugging this into the partial derivative, I obtain:
\begin{align*}
\frac{\p c^\ast(a,y;r,\abar)}{\p r} &= a - \frac{\rho}{\gamma r^2}(1+w(a)) - \frac{\rho-r}{\gamma r}\left[\frac{w}{f(a)(1+w(a))}f(a)(a-\abar)\gamma r\frac{2\rho-r}{(\rho-r)^2}\right]\\
&= a - \frac{\rho}{\gamma r^2}(1+w(a)) - \frac{w(a)}{1+w(a)}\frac{2\rho-r}{\rho-r}(a-\abar),
\end{align*}
which completes the proof.

\section{Proof of Proposition \ref{prop:ac_au}}
\label{sec:app_ac_au}
The proof is a bit involved given that $\p c^\ast(a,y;r,\abar)/\p r$ turns out
not to be strictly monotonic on $(\abar,+\infty)$. Therefore, I prove this
result in a number of steps:
\begin{enumerate}
    \item prove that $\dfrac{\p c^\ast(\abar,y;r,\abar)}{\p r}<0$,
    \item prove that $\dfrac{\p c^\ast(a,y;r,\abar)}{\p r}$ has a unique
          minimum at some $a^{\min}>\abar$,
    \item prove that $\dfrac{\p c^\ast(a,y;r,\abar)}{\p r}$ is strictly decreasing on
          $(\abar, a^{\min})$ and strictly increasing for $a>a^{\min}$.
\end{enumerate}
Once these steps are established, the fact that
$\p c^\ast(a,y;r,\abar)/\p r > \p c^\ast(a,y;r)/\p r$ for all $a$
(see Proposition \ref{prop:dcdr_decomp}) immediately implies that the MPC in the constrained case
crosses zero before the unconstrained one does.
 
\subsection{Negative sign at the constraint}
 
Proposition \ref{prop:dcdr_decomp} establishes that
\begin{align*}
  \frac{\p c^\ast(a,y;r,\abar)}{\p r}
  = a - \rho\frac{1+w(a)}{\gamma r^2}
    -\frac{w(a)}{1+w(a)}\frac{2\rho-r}{\rho-r}(a-\abar).
\end{align*}
As $a\to\abar^+$ one has $w(a)\to -1^+$, so the third term involves the
indeterminate form $w(a)/(1+w(a))\times(a-\abar)\to(-\infty)\times 0$. Using the Lambert identity
$\ln(-w(a))+w(a)=-1-(a-\abar)\gamma r^2/(\rho-r)$, one can write:
\begin{align*}
  \frac{w(a)}{1+w(a)}(a-\abar)
  = \frac{w(a)}{1+w(a)}\cdot\frac{\rho-r}{\gamma r^2}\bigl[-1-w(a)-\ln(-w(a))\bigr].
\end{align*}
Focusing on terms in $w(a)$, one can write:
\begin{align*}
  -\frac{w(a)}{1+w(a)}\bigl[1+w(a)+\ln(-w(a))\bigr]
  = -w(a) - \frac{w(a)\ln(-w(a))}{1+w(a)}.
\end{align*}
As $w(a)\to -1^+$, the first term tends to $1$. The second term is of the form
$0/0$ form. Applying l'H\^opital's rule gives
\begin{align*}
  \lim_{w(a)\to -1^+}\frac{w(a)\ln(-w(a))}{1+w(a)}
  = \lim_{w(a)\to -1^+}\bigl[\ln(-w(a))+1\bigr]
  = \ln(1)+1 = 1.
\end{align*}
Therefore
\begin{align*}
  \lim_{w(a)\to -1^+}\left[-w(a) - \frac{w(a)\ln(-w(a))}{1+w(a)}\right] = 1-1 = 0,
\end{align*}
and the third term vanishes as $a\to\abar^+$.  The second term likewise
vanishes in the limit since $1+w(a)\to 0$.  Evaluating the remaining first term
at $a=\abar$ one gets:
\begin{align*}
  \frac{\p c^\ast(\abar,y;r,\abar)}{\p r} = \abar < 0.
\end{align*}
 
\subsection{Unique minimum}
 
Define $h(a)\defeq\p c^\ast(a,y;r,\abar)/\p r$.  The cross-derivative
\begin{align*}
  G(w(a)) \defeq h'(a)
  = \frac{\p^2 c^\ast(a,y;r,\abar)}{\p r\,\p a}
  = \frac{1}{1+w(a)}+\frac{2\rho-r}{\rho-r}\frac{w(a)}{(1+w(a))^3}\bigl[-1-w(a)-\ln(-w(a))\bigr]
\end{align*}
is a function of $w(a)\in(-1,0)$, which is strictly increasing in $a$. Letting $w(a)=-e^{-x}$, $x=rT>0$ and noting that $(1+w(a))^3=(1-e^{-x})^3>0$, I define $P(x)\defeq (1+w(a))^3G(w(a))$. It follows that $\mathrm{sign}(G(w(a)))=\mathrm{sign}(P(x))$, where
\begin{align*}
  P(x) \defeq e^x+e^{-x}-2-\beta(e^{-x}-1+x),
  \qquad \beta\defeq\frac{2\rho-r}{\rho-r}>2,
\end{align*}
I establish that $P$ changes sign exactly once
on $(0,+\infty)$, from negative to positive.
 
\medskip\noindent\textbf{Step 1: $P(x)<0$ for small $x>0$.}
Direct computation gives $P(0)=0$ and $P'(0)=0$.  The second derivative is
$P''(x)=e^x+(1-\beta)e^{-x}$, so $P''(0)=2-\beta<0$ (since $\beta>2$).
A Taylor expansion therefore gives
\begin{align*}
  P(x) = \frac{2-\beta}{2}\,x^2+O(x^3) < 0
  \quad\text{for all sufficiently small }x>0.
\end{align*}
 
\medskip\noindent\textbf{Step 2: $P$ has a unique global minimum at
$\tilde{x}=\ln(\beta-1)$.}
One can show that
$e^x P'(x)=Q(y)=(y-1)(y+1-\beta)$ with $y=e^x$.
Since $y>1$ for $x>0$, the sign of $Q(y)$ is the sign of $y+1-\beta$ and therefore $P'(x)<0$ for $x\in(0,\tilde{x})$, i.e. $P$ is strictly decreasing on that interval. Likewise, $P'(x)>0$ for $x>\tilde{x}$, and $P$ is strictly increasing on that interval. Hence $\tilde{x}=\ln(\beta-1)>0$ is the unique global minimum of $P$.
 
\medskip\noindent\textbf{Step 3: $P(\tilde{x})<0$.}
Since $P(x)<0$ for small $x>0$ and $P$ is strictly decreasing on $(0,\tilde{x})$,
it follows that $P(\tilde{x})<P(x)<0$ for all $x\in(0,\tilde{x})$.
 
\medskip\noindent\textbf{Step 4: $P(x)\to+\infty$ as $x\to+\infty$.}
Since $P(x)=e^x+e^{-x}-2-\beta(e^{-x}-1+x)\sim e^x$ as $x\to+\infty$.
 
\medskip\noindent\textbf{Step 5: $P$ crosses zero exactly once.}
$P(\tilde{x})<0$ and $P$ is strictly increasing on $(\tilde{x},+\infty)$ with
$P(x)\to+\infty$.  By the intermediate value theorem there exists a unique
$x^\ast>\tilde{x}$ such that $P(x^\ast)=0$, with $P(x)<0$ for
$x\in(0,x^\ast)$ and $P(x)>0$ for $x>x^\ast$.
 
\medskip\noindent\textbf{Consequence for $G(w(a))$ and $h(a)$.}
Since $\mathrm{sign}(G(w(a)))=\mathrm{sign}(P(x))$, the cross-derivative
$G(w(a))=h'(a)$ is strictly negative on $(\abar,a^{\min})$ and strictly positive
on $(a^{\min},+\infty)$, where $a^{\min}$ is the asset level corresponding to
$x^\ast$ via $w(a)=-e^{-x}$ and the Lambert identity.  Therefore $h(a)$ is
strictly decreasing on $(\abar,a^{\min})$ and strictly increasing on
$(a^{\min},+\infty)$, with a unique global minimum at $a^{\min}$.
  
\subsection{MPC decreasing, then increasing}
 
By steps 1 to 5 and the consequence established above, $h(a)=\p c^\ast/\p r$ is
strictly decreasing on $(\abar,a^{\min})$ and strictly increasing on
$(a^{\min},+\infty)$.  This completes the proof of items 1 to 3.
 
\subsection{Conclusion}
 
The MPC with the borrowing constraint $(i)$ starts at
$h(\abar)=\abar<0$, $(ii)$ decreases to its global minimum $h(a^{\min})$,
and $(iii)$ increases thereafter toward $+\infty$.  Pairing this with
Proposition \ref{prop:dcdr_decomp}, which establishes $h(a)>\p c_u^\ast/\p r$ for all $a$, it
follows that $h$ crosses zero at some $\overline{a}^c<\overline{a}^u$.
Uniqueness of the zero follows from the fact that $h$ has a unique minimum and
is strictly increasing after $a^{\min}$: once $h$ becomes positive it stays
positive. To show $\overline{a}^c>0$, note that at $a=0$:
\begin{align*}
  \frac{\p c^\ast(0,y;r,\abar)}{\p r}
  = -\rho\frac{1+w(a)}{\gamma r^2}
    +\frac{w(a)}{1+w(a)}\frac{2\rho-r}{\rho-r}\abar.
\end{align*}
The first term is strictly negative.  Since $w(a)\in(-1,0)$ and $\abar\leq 0$,
the second term is also non-positive.  Hence the MPC at $a=0$ is strictly
negative, so $\overline{a}^c>0$. 

\section{Proof of Proposition \ref{prop:F_CRRA}}
\label{sec:app_F_CRRA}
As with the CARA case, it will be useful to work with the savings function $s(a)$ which is such that $c'(a)=r-s'(a)$. Plugging this into equation \eqref{eq:ODE_CRRA}, one obtains after re-arranging:
\begin{align*}
s'(a)=r+\frac{\rho-r}{\gamma}\frac{ra+y-s(a)}{s(a)}    
\end{align*}
In contrast with the CARA case, this one is not autonomous given that $a$ appears independently of $s(a)$ on the right hand side. This is what will prevent a closed form solution. In order to transform it into something that will be easier to handle, define $\check{s}(a)\defeq s(a)/(ra+y)$ which is the savings \textit{rate}. Using the fact that $s'(a) = \check{s}'(a)(ra+y)+r\check{s}(a)$, one can then write:
\begin{align*}
\check{s}'(a)(ra+y)+r\check{s}(a)=r+\alpha(\rho-r)\frac{1-\check{s}(a)}{\check{s}(a)}.    
\end{align*}
Now multiplying both sides by $\check{s}(a)$, dividing by $r$ and re-arranging gives:
\begin{align}
\label{eq:scheck_a_factored}
\check{s}(a)\check{s}'(a)\left(a+\frac{y}{r}\right) &= -\check{s}(a)^2 + (1-\alpha\theta)\check{s}(a) + \alpha\theta\\
&=(1-\check{s}(a))(\check{s}(a)+\alpha\theta)
\end{align}
Now consider the following change of variables: $\tau(a)=\ln\left\{ra+y\right\}$. Using the chain rule, one can write $d\check{s}(a)/da = (d\check{s}(a)/d\tau)(d\tau/da)=(d\check{s}(a)/d\tau)(r/(ra+y))$ which in turn implies that $(ra+y)(d\check{s}(a)/da)=r(d\check{s}(a)/d\tau)$. With this, equation \eqref{eq:scheck_a_factored} becomes:
\begin{align}
\nonumber
(1-\check{s}(a))(\check{s}(a)+\alpha\theta) &= \check{s}(a)\frac{d\check{s}(a)}{d\tau}\\
\Rightarrow\quad \frac{\check{s}(a)}{(1-\check{s}(a))(\check{s}(a)+\alpha\theta)}d\check{s}(a) &= d\tau
\label{eq:dscheck_a_dtau}
\end{align}
Using partial fractions, one can write
\begin{align*}
\frac{\check{s}(a)}{(\check{s}(a)-1)(\check{s}(a)+\alpha\theta)} = \frac{A_1}{\check{s}(a)-1} + \frac{A_2}{\check{s}(a)+\alpha\theta},    
\end{align*}
where $A_1=1/(1+\alpha\theta)$ and $A_2=\alpha\theta/(1+\alpha\theta)$. Using this to integrate from $\check{s}(\abar)=0$ to a generic $\check{s}(a)$ for the left hand side and from $\ln(c_{min})$ to $\ln(ra+y)$, one obtains:
\begin{align*}
\int_{0}^{\check{s}(a)}\left(\frac{A_1}{1-x}-\frac{A_2}{x+\alpha\theta}\right)dx  &= -A_1\ln(1-\check{s}(a))-A_2\ln(\check{s}(a)+\alpha\theta)+A_2\ln(\alpha\theta)\\
&= \ln(ra+y)-\ln(c_{min}).
\end{align*}
Multiplying by $1+\alpha\theta$, undoing the change of variables and re-arranging, one gets:
\begin{align*}
\ln\left\{\frac{c(a)}{ra+y}\right\}+\ln\left\{\left(\frac{(1+\alpha\theta)(ra+y)-c(a)}{\alpha\theta(ra+y)}\right)^{\alpha\theta}\right\}=\ln\left\{\left(\frac{c_{min}}{ra+y}\right)^{1+\alpha\theta}\right\}    
\end{align*}
Exponentiating both sides and re-arranging, one finally obtains
\begin{align*}
c(a) = (1+\alpha\theta)(ra+y)-\alpha\theta c_{min}\left(\frac{c_{min}}{c(a)}\right)^{\frac{1}{\alpha\theta}},
\end{align*}
which is the definition of $F(c(a),a;r)$. This completes the proof.

\section{Proof of Theorem \ref{thm:dcdr_CRRA}}
\label{sec:app_CRRA}

Keeping $r$ fixed, the partial derivative of $F$ with respect to consumption is given by
\begin{align*}
\frac{\p F}{\p c} &= 1 - \alpha\theta\kappa\left(\frac{c}{c_{min}}\right)^{-(\kappa+1)}    \\
&= 1 - x^{-(\kappa+1)}
\end{align*}
where $x\defeq c/c_{min}$ and I have used the fact that $\kappa=1/(\alpha\theta)$. Given that $0<c_{min}< c$, it follows that $0<x< 1$ for $a\in(\abar,\infty)$. As a result, the partial derivative is guaranteed to be strictly positive. Therefore, the sign of $\p c(a;r)/\p r$ is the same as the one for $-F_r$, which, using the chain rule is given by:
\begin{align*}
-F_r = (1+\alpha\theta)(a-\abar x^{-\kappa}) - \frac{\alpha\rho}{r^2}(r&+y-c_{min}x^{-\kappa})+\frac{\rho}{r(\rho-r)}c_{min}\ln(c)x^{-\kappa}    
\end{align*}
Along the solution path, we have $F=0$ which one can use to express $a$ as a function of $c$ so that $-F_r$ becomes a function of $x$ only. Indeed, $F=0$ implies that 
\begin{align*}
a = \frac{1}{r}\left[c_{min}\frac{x+\alpha\theta x^{-\kappa}}{1+\alpha\theta}-y\right].    
\end{align*}
Plugging that into the equation for $-F_r$ and re-arranging, one obtains:
\begin{align*}
-F_r = \frac{c_{min}}{r}\left\{Ax-B+\left[B-A+C\ln(x)\right]x^{-\kappa}\right\},    
\end{align*}
where $A=(\gamma-1)/(\gamma+\theta)$, $B=(1+\alpha\theta)y/c_{min}$ and $C=\rho/(\rho-r)$. It follows that studying the sign of $\p c(a;r)/\p r$ boils down to studying the sign of the term in curly brackets on the right hand side, which we define as $n(x)$ given that it sits at the numerator. First of all, notice that its derivative is given by:
\begin{align}
\nonumber
n'(x) &= A + x^{-(\kappa+1)}\left[C-\kappa(B-A+C\ln(x))\right]\\
\label{eq:nprime_x}
&\defeq A+x^{-(\kappa+1)}\left[D-\kappa C\ln(x)\right].
\end{align}
For later, it will be useful to compute the value of $n'(x)$ for $c=c_{min}$:
\begin{align*}
n'(1) = A+D=(1+\kappa)\frac{r\abar}{c_{min}}<0.    
\end{align*}
where I have used Assumption \ref{ass:parameters}. Letting $t=\ln(x)$ and $u=D-\kappa C \cdot t$, equation \eqref{eq:nprime_x} is equivalent to:
\begin{align*}
n'(x) = A+\frac{e^{-\nu D}}{\nu}(\nu u)e^{\nu u}\quad\text{so}\quad n'(x)<0 \Leftrightarrow (\nu u)e^{\nu u}<\zeta\defeq -A\nu e^{\nu D}     
\end{align*}
where I have defined $\nu=(1+\kappa)/(\kappa C)$. The last equation clearly lends itself to the Lambert W function, which means that one has to be careful about branches. To that effect, let me first consider the case where $\gamma>1$ so that the EIS is less than one and $A>0$. In that case, the right hand side of the last equation is strictly negative and the endpoints for the interval of $x$ are given by the two branches of the Lambert W function:
\begin{align*}
\left\{x:n'(x)<0\right\} = (x_-,x_+)\quad \text{where}\quad x_- = exp\left\{\frac{\nu D-W_0(\zeta)}{\nu\kappa C}\right\}; x_+ = exp\left\{\frac{\nu D-W_{-1}(\zeta)}{\nu\kappa C}\right\}     
\end{align*}
It has been established already that $n'(1)<0$ so that $(x_-,x_+)$ is non-empty and contains 1. In particular, $x_-<1$ and therefore $n(x)$ strictly decreases on $(1,x_+)$. Given that $n(1)<0$, $n(x)<0$ on that interval. By definition one can write that $n'(x_+)=0$. In addition, using the chain rule one can compute the second derivative of $n(x)$ as follows:
\begin{align*}
n''(x)=-\kappa C x^{-(\kappa+2)}(1+\nu u).    
\end{align*}
For all $x\geq x_+$ we have $\nu u\leq W_{-1}(\zeta)\leq -1$ by definition. It follows that $n(x)$ is convex on $(x_+,\infty)$. Putting it all together, $n'(x_+)=0$ and $n''(x)> 0$ for $x> x_+$ so that $n'(x)> 0$ on the same interval. It follows that $n(x)$ is strictly increasing on $(x_+,\infty)$ and starting from a value $n(x_+)< 0$. Using the standard fact that power beats log, it is straightforward to show that $n(x)\to\infty$ as $x\to\infty$ because of the linear $Ax$ term. As a result, there exists a single point where $n(x)$ crosses 0. Equivalently, there exists $\overline{a}^{CRRA}\in (\abar,\infty)$ such that $\p c^*(a;r)/\p r> 0$ if and only if $a>\overline{a}^{CRRA}$. This proves the second bullet point of Theorem \ref{thm:dcdr_CRRA}. 

To prove the first point, note that if $\gamma\leq 1$ so that the EIS is larger than 1 then we have $\zeta\geq 0$. In that case, the interval for which $n'(x)<0$ is given by $[x_0,\infty)$ where \begin{align*}
x_0 \defeq \exp\left\{\frac{\nu D-W_{0}(\zeta)}{\nu\kappa C}\right\}.     
\end{align*}
Given that $n'(1)<0$ regardless of whether $\zeta\lessgtr 0$, it follows that $n'(x)<0$ on $[1,\infty)$. Given that $n(1)<0$ and that $n(x)$ is strictly decreasing on $[1,\infty)$, $F_r$ is always positive so that consumption necessarily decreases as soon as the interest rate increases. This proves the first bullet point of Theorem \ref{thm:dcdr_CRRA}.

\end{document}